%% file: paper.tex
\documentclass[sigconf]{acmart}

\setcopyright{none} 
\acmYear{2021}

\usepackage{algorithmic}
\usepackage{ifthen}
\usepackage{amsmath,amsfonts}
\usepackage{algorithmic}
\usepackage{graphicx}
\usepackage{textcomp}
\usepackage{xcolor}
\usepackage{colortbl}
\def\BibTeX{{\rm B\kern-.05em{\sc i\kern-.025em b}\kern-.08em
    T\kern-.1667em\lower.7ex\hbox{E}\kern-.125emX}}
\usepackage{tikz}
\usetikzlibrary{arrows, calc, fit, positioning, shapes}
\usepackage{amsthm}
\usepackage{thmtools, thm-restate}
\declaretheorem[name=Theorem, style=theorem]{Theorem}
\declaretheorem[name=Lemma, style=theorem]{Lemma}

\usepackage{mathpartir}
\usepackage{todonotes}
\presetkeys{todonotes}{inline}{}
\makeatletter
\providecommand{\bigsqcap}{%
  \mathop{%
    \mathpalette\@updown\bigsqcup
  }%
}
\newcommand*{\@updown}[2]{%
  \rotatebox[origin=c]{180}{$\m@th#1#2$}%
}
\makeatother
\usepackage{tikz-cd}
\usepackage{pgfplots}
\usepackage{pgfplotstable}
\usepackage{adjustbox}
\usepackage{bashful}
\usepackage{siunitx}
\usepackage{mathtools}

\newcommand{\LabelPrefix}{prefix}
\newcommand{\appendixsection}[1]{\section{#1}}

\input{syntax.tex}

\begin{document}



\title{Multi-Execution Lattices Fast and Slow
}


\begin{abstract}
Methods for automatically, soundly, and precisely guaranteeing the
noninterference security policy are predominantly based on multi-execution.
All other methods are either based on undecidable theorem proving or
suffer from false alarms.
The multi-execution mechanisms, meanwhile, work by isolating security levels
during program execution and running multiple copies of the target program,
once for each security level with carefully tailored inputs that ensure both
soundness and precision.
When security levels are hierarchically organised in a lattice, this may lead
to an exponential number of executions of the target program as the number
of possible ways of combining security levels grows.
In this paper we study how the lattice structure for security levels
influences the runtime overhead of multi-execution.
We additionally show how to use Galois connections to gain speedups in
multi-execution by switching from lattices with high overhead to lattices with
low overhead.
Additionally, we give an empirical evaluation that corroborates our analysis
and shows how Galois connections have potential to speed up multi-execution.
\end{abstract}

\author{Maximilian Algehed}
\affiliation{%
   \institution{Chalmers}
   \city{Gothenburg}
   \country{Sweden}}
\email{m.algehed@gmail.com}
\author{Cormac Flanagan}
\affiliation{%
   \institution{UC Santa Cruz}
   \city{Santa Cruz}
   \state{CA}
   \country{USA}}
\email{cormac@ucsc.edu}

\maketitle

\input{content.tex}


\bibliographystyle{ACM-Reference-Format}
\bibliography{main.bib}

\appendix

\input{appendix.tex}

\end{document}

%% file: syntax.tex
\newcommand{\Nat}{\mathbb{N}}
\newcommand{\Q}{\mathbb{Q}}
\newcommand{\ME}{\text{ME}}
\newcommand{\MEF}{\text{MEF}}

\newcommand{\project}{\operatorname{\downarrow}}
\renewcommand{\L}{\mathcal{L}}

\newcommand{\Pow}{\mathcal P}

\renewcommand{\iff}{\Leftrightarrow}
\renewcommand{\implies}{\Rightarrow}

\newcommand{\canFlowTo}{\sqsubseteq}

\newcommand{\lub}{\sqcup}

\newcommand{\Alice}{\text{Alice}}
\newcommand{\Bob}{\text{Bob}}
\newcommand{\Charlie}{\text{Charlie}}
\newcommand{\Dave}{\text{Dave}}
\newcommand{\LOW}{\texttt{L}}
\newcommand{\MED}{\texttt{M}}
\newcommand{\HIGH}{\texttt{H}}

\newcommand{\nameify}[1]{\textit{#1}}

\newcommand{\defAs}{\ensuremath{\triangleq}}

\newcommand{\CS}[1]{\textit{CS}_{#1}}
\newcommand{\Oh}{\textit{O}}
\newcommand{\maxi}{\textit{max}}
\DeclareMathOperator{\vsum}{\uparrow}
\DeclareMathOperator{\hsum}{\diamond}

\newcommand{\badSum}{\nameify{badSum}}
\newcommand{\goodSum}{\nameify{goodSum}}
\newcommand{\pairwise}{\nameify{pairwise}}

\newcommand{\specify}{\nameify{specify}}
\newcommand{\unspecify}{\nameify{unspecify}}
\newcommand{\truncate}{\nameify{truncate}}
\newcommand{\embed}{\nameify{embed}}
\newcommand{\secure}{\nameify{secure}}
\newcommand{\insecure}{\nameify{insecure}}
\newcommand{\badListSum}{\nameify{badListSum}}
\newcommand{\goodListSum}{\nameify{goodListSum}}
\newcommand{\discrete}{\mathcal{D}}
\newcommand{\Fac}{\textit{Fac}}
\newcommand{\facet}[3]{\langle\ #1\ ?\ #2\ :\ #3\ \rangle}

%% file: content.tex
\section{Introduction}
\label{\LabelPrefix:sec:intro}

Language-based information flow control (IFC) is a set of techniques for
controlling the way that information flows inside a program
\cite{sabelfeld2003language}.
The techniques in this field are generally aimed at ensuring the
\emph{noninterference} security policy \cite{goguen1982security,Denning}: a
program $p$ is \emph{noninterfering} if its secret inputs cannot influence its
public outputs.
Traditional enforcement mechanisms for IFC, whether static \cite{MAC, JIF,
FlowCaml} or dynamic
\cite{InlineMonitors,LIO,hedin2014jsflow,austin2010permissive}, suffer either
from high false-alarm rates \cite{King, staicu2019empirical} or undecidability
\cite{barthe2011relational,barthe2011secure}.

To address this issue, IFC researchers have introduced a family of enforcement
mechanisms collectively known as \emph{multi-execution} \cite{SME, MF, FSME,
OGMF, jaskelioff2011secure}.
Multi-execution guarantees \emph{transparency}: if $E$ is a \emph{transparent}
enforcement mechanism and $p$ is a noninterfering program, then $E[p](x) = p(x)$
for all $x$.
Consequently, multi-execution cannot introduce any false alarms.

How it works is in the name, programs are executed multiple times to ensure
that output at each security level can only depend on input visible to that
level.
A consequence of this strategy is that in the worst case, multi-execution
introduces unmanageable performance overhead \cite{OptimisingFSME,
MultiExecutionBounds, FSME, MF}.
In fact, under certain assumptions, it is impossible to construct a transparent
enforcement mechanism that introduces less than exponential worst-case overhead
in execution time on secure programs \cite{MultiExecutionBounds}.

To be secure and transparent, multi-execution has to respect every way a
program combines data of different security levels \cite{MultiExecutionBounds}.
If the program $p$ combines data sensitive to Alice, Bob, and Charlie,
multi-execution is forced to run $p$ eight times, once for every subset of
$\{\Alice, \Bob, \Charlie\}$.
However, if $p$ only combines public and secret data, multi-execution only
needs to run $p$ twice, once for secret and once for public data, as the
combination of public and secret data is secret data.

In general, we consider the case where the security levels are drawn from a
lattice \cite{Denning}.
As demonstrated above, the shape of the lattice matters.
A powerset lattice over principals (like Alice, Bob, and Charlie)
introduces an exponential number of combinations of security levels, but a
total order (like the two-point lattice with levels for public and secret
data) doesn't suffer from this issue.
This relationship is the core of what this paper is about.

There are a number of practical takeaways from this insight.
Applications that do not require the combinatorial power of powerset
lattices do well under multi-execution.
For example, Alpernas et al. \cite{alpernas2018secure} use three lattices in
their case studies of IFC in a serverless setting that are all ``wide and
short''.
In their lattices, each user has a security label, but the applications do not
distinguish between arbitrary combinations of users, and the lattices have a
flat structure.
For example, their \texttt{gg} lattice has least and greatest elements $\bot$
and $\top$ and a set of incomparable elements $\{\Alice, \Bob, \ldots\}$ where
the combination of any two such elements is $\top$.
At most $n+2$ executions need to be performed when multi-executing in this
lattice, one execution for each principal that appears in the input, one for
$\bot$, and one for $\top$.
This demonstrates that some applications admit lattices that accommodate the
security requirements of the application while introducing low overhead.

The primary goal of this paper, then, is to answer the following question:
Given some lattice $\L$ and input elements that are associated with $N$
distinct labels in $\L$, what is the time overhead that multi-execution imposes
with respect to $N$?

To answer this question, we describe a new perspective on multi-execution that
comes in two parts.
Firstly, we provide the tools necessary to analyse security lattices to
quantify how much overhead they introduce in multi-execution.
Secondly, we show how to pair the fact that some lattices introduce less
overhead than others with the insight that Galois connections give rise to
natural translations between security lattices in multi-execution.

Concretely, we make the following contributions:
\begin{enumerate}
  \item We characterise the connection between the choice of security lattice
    and the worst-case runtime of black-box multi-execution enforcement
    (Sections \ref{sec:greatAndSmall} and \ref{sec:fastAndSlow}).
  \item We present a theory for computing bounds on multi-execution
    overhead for different lattices (Section \ref{sec:greatAndSmall}).
  \item We show how Galois connections reduce the overhead of multi-execution
    by executing in one lattice while observing the results in another (Section
    \ref{sec:lookingGlass}).
  \item We give a method for specifying optimal Galois connections for
    multi-execution (Section \ref{sec:lookingGlass}).
  \item We present a Haskell implementation of our techniques and empirically evaluate
    our predictions (Section \ref{sec:empirical}).
\end{enumerate}

\section{Review of the Multi-Execution Framework}
\label{\LabelPrefix:sec:prelim} 

A (join semi-)lattice $\L$ is a set $\L$ with a transitive, reflexive,
and antisymmetric order $\canFlowTo$ that has a least element $\bot$ and
is such that any two elements $\ell, \jmath \in \L$ have a least
upper bound $\ell \lub \jmath \in \L$.
For a finite subset $S \subseteq \L$ we write $\bigsqcup S$ for the least upper bound
of all elements in $S$.
For example, the \emph{two-point} lattice has $\L = \{ \LOW, \HIGH \}$, $\LOW$
denotes public information and $\HIGH$ denotes secret information.
Public information can flow to secret information so $\canFlowTo$ is the
smallest reflexive relation such that $\LOW \canFlowTo \HIGH$.
Finally, this means that $\LOW \lub \LOW = \LOW$ and $\ell \lub \jmath = \HIGH$ if
either $\ell$ or $\jmath$ is $\HIGH$.

Following Algehed and Flanagan \cite{MultiExecutionBounds} we consider
batch-job programs from labeled sets to labeled sets and let $p, q, r$ range
over partial recursive functions from $\Pow(I \times \L)$ to $\Pow(O \times \L)$
for some set of inputs $I$ and outputs $O$.
This is a convenient formalism, as it allows us to succinctly state the core
definitions that allow us to reason about multi-execution.
The following definitions (from \cite{MultiExecutionBounds}) are sufficient to
precisely define Noninterference.

\begin{definition}
  Assume $\ell \in \L$ and $x, y \subseteq V \times \L$ for some $V$, define \emph{the projection
  of $x$ at $\ell$} as (we write the pair $(a, \jmath)$ as $a^\jmath$):
  $$x \project \ell \defAs \{\ a^{\jmath}\ |\ a^{\jmath} \in x, \jmath \canFlowTo \ell \}$$
  We say that $x$ and $y$ are \emph{$\ell$-equivalent}, meaning they look the
  same to an observer at level $\ell$, written $x \sim_\ell y$, if and only if
  their $\ell$-projections are the same:
  $$x \sim_\ell y\ \iff\ x\project\ell = y\project\ell$$
\end{definition}

The projection $x\project\ell$ of $x$ at $\ell$ is precisely all the
information in $x$ that is \emph{visible to} $\ell$.
Likewise, this means that if two sets $x$ and $y$ look the same to $\ell$, then
they are $\ell$-equivalent.
The definition of noninterference meanwhile is that $p$ is noninterfering if it
does not reveal more about its inputs than what one can know by looking at the
input.
In other words, if two inputs $x$ and $y$ differ only in values that are secret
to an observer at level $\ell$, they are $\ell$-equivalent, then $p(x)$ and
$p(y)$ should also be $\ell$-equivalent.

\begin{definition}[Noninterference]
  We say that program $p : \Pow(I\times\L) \to \Pow(O\times\L)$ is \emph{noninterfering}
  if it preserves $\ell$-equivalence.
  Concretely, $p$ is noninterfering when for all $\ell$, $x$, and $y$ such that
  \begin{enumerate}
    \item $x \sim_\ell y$ and
    \item $p(x)$ and $p(y)$ are both defined,
  \end{enumerate}
  it is the case that $p(x) \sim_\ell p(y).$
\end{definition}

Note that this definition of noninterference is a partial correctness criterion,
it says that $\ell$-equivalence only has to be preserved up to termination,
known as \emph{Termination Insensitive} Noninterference (TINI)
\cite{hedin2012perspective}.
The theory of termination sensitivity in this setting is rich \cite{MultiExecutionBounds}.
However, termination is orthogonal to our development and we omit it here.

\begin{example}
  The program $\secure$ below is noninterfering (we write $|x|$
  for the size of the set $x$):
  \begin{align*}
    \secure(x) &\defAs \{\ |x\project\ell|^\ell\ |\ \ell \in \{\LOW, \HIGH\}\ \}\\
    \intertext{Conversely, the program $\insecure$ is not noninterfering:}
    \insecure(x) &\defAs \{\ |x|^\LOW\ \}
  \end{align*}
\end{example}

Before we dive into more examples of how this framework works we
introduce a core lattice for this paper, the \emph{powerset lattice} $\Pow(A)$
over some set $A$ of \emph{atoms} or \emph{principals}.
A label $\ell$ in $\Pow(A)$ is a subset $\ell \subseteq A$ and labels 
are ordered by set inclusion, $\ell \canFlowTo \ell'$ if and only if $\ell
\subseteq \ell'$.
The least element of $\Pow(A)$ is the empty set and the least
upper bound of two labels $\ell \lub \ell'$ is their union $\ell \cup \ell'$.
We usually write singleton labels, like $\{\Alice\}$,  without the brackets as $\Alice$.
Finally, given $x \in \Pow(V \times \L)$ we define the \emph{labels} of $x$ as:
$$
\L(x) \defAs \{\ \ell\ |\ a^\ell \in x \}
$$

\begin{example}
  \label{ex:runnings}
  We present our running examples. 
  First is the program $\badSum$, that takes the sum of its inputs
  and labels the output with the least upper bound of the labels in the
  input.
  \begin{align*}
    \badSum(x) &\defAs \{ (\sum_{a^\ell \in x} a)^{\bigsqcup \L(x)} \}\\
    \intertext{This program is not noninterfering, as $\emptyset \sim_\bot \{1^\Alice\}$ but:}
    \badSum(\emptyset) &= \{ 0^\bot \} \not\sim_\bot \{ 1^\Alice \} = \badSum(\{1^\Alice\})
    \intertext{The problem is that the definition of noninterference is
    \emph{presence-sensitive}; input at a label is considered sensitive
    information that may not leak.
    To address the problem with $\badSum$ we define $\goodSum_L$, a family of
    programs indexed by a set of labels $L$ that are required to be in the
    input, and which form the levels for which the sum is taken; and consequently $\goodSum_L$
    is noninterfering:}
    \goodSum_L(x) &\defAs \{\ (\sum \{\ a\ |\ a^\ell \in x, \ell \in L\ \})^{\bigsqcup L}\ \}\\
    \intertext{The next noninterfering program combines input data pairwise.}
    \pairwise(x) &\defAs \{\ \textit{max}(a, b)^{\ell \lub \jmath}\ |\  a^\ell \in x, b^\jmath \in x\ \}\\
  \end{align*}
\end{example}

We want enforcement mechanisms for noninterference not to alter the semantics
of programs that are already noninterfering.
This type of enforcement is known as \emph{transparent} IFC enforcement
\cite{zanarini2013precise}, and has been extensively studied in the literature
\cite{SME,MF,jaskelioff2011secure,FSME,OGMF,OptimisingFSME,MultiExecutionBounds,
zanarini2013precise,de2012flowfox,de2014secure,ngo2015runtime,pfeffer2019efficient,
micinski2020abstracting,rafnsson2016secure,bolocsteanu2016asymmetric}.
The common denominator of all these is that they are based on the idea of
multi-execution \cite{SME}.
Figure \ref{fig:sme} (originally appearing in \cite{MultiExecutionBounds})
illustrates multi-execution in the setting of public and secret data.
Multi-execution runs the program $p$ twice to produce $\ME[p]$, once with only
public input (this is called the ``public run'') and once with both public and
private input (this is called the ``private run'').
The final public outputs come from the public run, and the private outputs
from the private run.

This guarantees noninterference; the output in the public run cannot
depend on the secret input.
Similarly, if $p$ is noninterfering then multi-execution preserves its
extensional behaviour.
The secret output of $\ME[p]$ is the same as the secret output of $p$,
and the public output of $\ME[p]$ is the same as the public output of $p$
by virtue of $p$ being noninterfering.

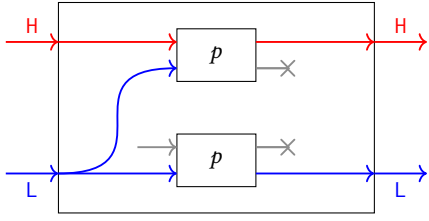
\begin{figure}[t]
  \centering
\begin{tikzpicture}[scale = 0.7]
  \draw (-3, -2) rectangle (3, 2);      

  \draw [line width = 0.25mm, red,  ->] (-4,  1.25) -- node [above] {\HIGH} (-3,  1.25); 
  \draw [line width = 0.25mm, red,  ->] (3,   1.25) -- node [above] {\HIGH} (4,   1.25); 

  \draw [line width = 0.25mm, blue, ->] (-4, -1.25) -- node [below] {\LOW}  (-3, -1.25); 
  \draw [line width = 0.25mm, blue, ->] (3,  -1.25) -- node [below] {\LOW}  (4,  -1.25); 

  \draw (-0.75,  0.5) rectangle (0.75,  1.5) node[pos=0.5]{$p$}; 
  \draw (-0.75, -0.5) rectangle (0.75, -1.5) node[pos=0.5]{$p$}; 

  \draw [line width = 0.25mm, red,  ->] (-3,  1.25) -- (-0.75,  1.25); 
  \draw [line width = 0.25mm, blue, ->] (-3, -1.25) -- (-0.75, -1.25); 

  \draw [line width = 0.25mm, blue, ->] (-3, -1.25) .. controls (-0.75, -1.25) and (-3,  0.75) .. (-0.75, 0.75); 
  \draw [line width = 0.25mm, black!45, ->] (-1.5, -0.75) -- (-0.75, -0.75);

  \draw [line width = 0.25mm, red,  ->] (0.75,  1.25) -- (3,  1.25); 
  \draw [line width = 0.25mm, blue, ->] (0.75, -1.25) -- (3, -1.25); 

  \draw [line width = 0.25mm, black!45, >=angle 90, ->] (0.75, -0.75) -- (1.40, -0.75); 
  \draw [line width = 0.25mm, black!45, >=angle 90, -<] (0.75, -0.75) -- (1.5, -0.75); 
  \draw [line width = 0.25mm, black!45, >=angle 90, ->] (0.75, 0.75)  -- (1.40,  0.75); 
  \draw [line width = 0.25mm, black!45, >=angle 90, -<] (0.75, 0.75)  -- (1.5,  0.75); 
\end{tikzpicture}
  \caption{\normalsize \label{fig:sme} Secure Multi-Execution of the program $p$ for the two-point lattice.}
\end{figure}

To formalise multi-execution we first re-state Algehed and Flanagan's definitions of some
auxiliary functions:
\begin{definition}
  Given a finite subset $S \subseteq \L$ of $\L$ we define the \emph{closure set} of $S$ as:
  \begin{align*}
    C(S) &\defAs \{\ \bigsqcup S'\ |\ S' \subseteq S\ \}\\
  \intertext{This is the set of all combinations of levels in $S$, and as we
    have seen it corresponds to the runs that multi-execution will have to do
    when $S$ are the labels in the input to the program.
    The next notion we define is the \emph{up-set} of $\ell$ in $S$ as:}
    \ell \uparrow S &\defAs \{\ \jmath\ |\ \ell \canFlowTo \jmath,
      \forall \iota \in S.\ \iota \canFlowTo \jmath \implies \iota \canFlowTo \ell \}\\
  \intertext{This is a technical notion that captures all the labels that ``see
    the same view'' of an input.
    It is what allows multi-execution to correctly propagate outputs from the
    target program that are not strictly combinations of the security levels in
    the input.
    Finally, given an $x \subseteq V \times \L$ and an $L \subseteq \L$ we
    define the \emph{selection} of $x$ at $L$:}
    x@L &\defAs \{\ a^\ell\ |\ a^\ell \in x, \ell \in L \}
  \end{align*}
\end{definition}
Next we formalise multi-execution:
\begin{definition}[Multi-Execution \cite{MultiExecutionBounds}]
  $$\MEF[p](x) = \bigcup\{\ p(x\project\ell)@(\ell\uparrow C(\L(x)))\ |\ \ell \in C(\L(x))\ \}$$
\end{definition}

The definition of $\MEF$ is superficially different from the overview in Figure
\ref{fig:sme}.
Specifically, $\MEF[p](x)$ runs $p$ for every $\ell$ in $C(\L(x))$ rather
than $\L$.
This is because the set of all projections $x\project\ell$ for $\ell \in \L$
is equal to the set of all projections $x\project\ell$ for $\ell \in C(\L(x))$.
Consequently, $\MEF[p](x)$ runs $p(x\project\ell)$ for all the $\ell$ necessary
to have every ``view'' of $x$.

The $\uparrow$ construction is responsible for reconstructing \emph{the
outputs} at each level in $\L$.
The intuition for this construction is that $\ell \uparrow C(\L(x)))$ is the
set of levels $\{\jmath_1, \ldots, \jmath_n\}$ such that the execution of
$p(x\project\ell)$ is responsible for computing the output at levels
$\jmath_i$.
Formally, $\ell \uparrow C(\L(x))$ is the set of all $\jmath_i$ such that
$x\project\jmath_i = x\project\ell$.
Thus any output of $p(x\project\ell)$ that is labeled $\jmath_i$ would have been
present in $p(x\project\jmath_i)$ and is therefore ``safe'' to include in the output
of $\MEF[p](x)$.

For example, consider what happens when we run
$$\MEF[\goodSum_L](\{1^\Alice, 2^\Charlie\})$$
for $L = \{\Alice, \Bob\}$.
We have that
$$C(\{\Alice, \Charlie\}) = \{\bot, \Alice, \Charlie, \{\Alice, \Charlie\}\}.$$
This means that we have four runs of $\goodSum$:
\begin{align*}
  &\goodSum_L(\{1^\Alice, 2^\Charlie\}\project\bot)                 = \{ 0^{\{\Alice, \Bob\}} \}\\
  &\goodSum_L(\{1^\Alice, 2^\Charlie\}\project\Alice)               = \{ 1^{\{\Alice, \Bob\}} \}\\
  &\goodSum_L(\{1^\Alice, 2^\Charlie\}\project\Charlie)             = \{ 0^{\{\Alice, \Bob\}} \}\\
  &\goodSum_L(\{1^\Alice, 2^\Charlie\}\project\{\Alice, \Charlie\}) = \{ 1^{\{\Alice, \Bob\}} \}
\end{align*}
To determine the final output of $\MEF[\goodSum_{\{\Alice, \Bob\}}]$, we need
to decide which of these outputs we preserve to the final output.
Figure \ref{fig:up-sets} shows the up-sets of $\bot$, $\Alice$, $\Charlie$, and
$\{\Alice, \Charlie\}$ in the powerset lattice for three principals $\Alice$,
$\Bob$, and $\Charlie$ when $\L(x) = \{\Alice, \Charlie\}$.
We see that $\{\Alice, \Bob\}$ falls in the up-set of $\Alice$, and so we have that
the final result is:
$$\MEF[\goodSum_{\{\Alice, \Bob\}}](\{1^\Alice, 2^\Charlie\}) = \{ 1^{\{\Alice, \Bob\}} \}$$
\begin{figure}
  \center
  \begin{tikzcd}
    [execute at end picture={
      \draw[draw=black, thick, rounded corners, dashed]
        (A.north west) -- (AB.north west) -- (AB.north east) -- (AB.south east) -- (A.south east) -- (A.south west) -- cycle;
      \draw[draw=black, thick, rounded corners, dashed]
        (AC.south west) -- (ABC.north west) -- (ABC.north east) -- (AC.south east) -- cycle;
      \draw[draw=black, thick, rounded corners, dashed]
        (bot.south west) -- (B.north west) -- (B.north east) -- (bot.south east) -- cycle;
      \draw[draw=black, thick, rounded corners, dashed]
        (C.north east) -- (BC.north east) -- (BC.north west) -- (BC.south west) -- (C.south west) -- (C.south east) -- cycle;
    }]
    & & |[alias=ABC]|ABC \arrow[ld, no head] \arrow[d, no head] \arrow[rd, no head] & & \\
    & |[alias=AB]|AB & |[alias=AC]|AC \arrow[lld, no head] \arrow[rrd, no head] & |[alias=BC]|BC \arrow[ld, no head] \arrow[rd, no head] & \\
    |[alias=A]|A \arrow[ru, no head] & & |[alias=B]|B \arrow[lu, no head] & & |[alias=C]|C \\
    & & |[alias=bot]|\bot \arrow[llu, no head] \arrow[rru, no head] \arrow[u, no head] & &
  \end{tikzcd}
  \caption{\label{fig:up-sets} A graphical representation of the $\ell\uparrow C(\{\Alice, \Charlie\})$ sets in $\Pow(\{\Alice, \Bob, \Charlie\})$ for $\ell \in \{\bot, \Alice, \Charlie, \{\Alice, \Charlie\}\}$.}
\end{figure}
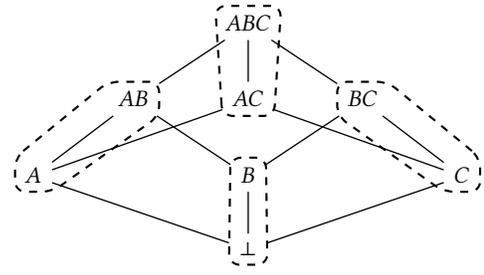

$\MEF$ enjoys both noninterference and transparency \cite{MultiExecutionBounds}.
\begin{Theorem}[Security]
  $\MEF[p]$ is noninterfering.
\end{Theorem}
\begin{Theorem}[Transparency]
  \label{thm:MEF-transparent}
  If $p$ is noninterfering, then:
  $$\MEF[p](x) = p(x)$$
\end{Theorem}

Note that $\MEF[p](x)$ works by ``looping over'' $C(\L(x))$.
In general, $|C(\L(x))|$ can be exponential in $|\L(x)|$ and so the runtime overhead
of $\MEF$ is substantial.
In the worst case with a powerset lattice, $\MEF[p](x)$ is exponentially
slower than $p(x)$, even when $p$ is a noninterfering program.

More generally, the runtime of $\MEF[p](x)$ is bounded by the runtime of $p$
times the size of $C(\L(x))$ (and some auxiliary computations that we return to
in Section \ref{sec:fastAndSlow}).
Consequently, the choice of lattice $\L$ makes a big difference to performance
as it decides the number of executions of $p$.

\begin{figure}
  \centering
  \begin{tikzcd}
    & \top & \\
    \Alice \arrow[ru, no head] & \Bob \arrow[u, no head] & \Charlie \arrow[lu, no head] \\
    & \bot \arrow[lu, no head] \arrow[u, no head] \arrow[ru, no head] &
  \end{tikzcd}
  \caption{\label{fig:discrete} The Discrete Lattice $\discrete(\{\Alice, \Bob, \Charlie\})$}
\end{figure}
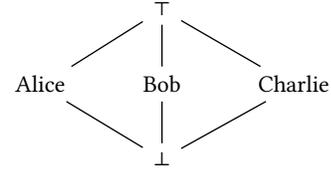

\begin{example}
  \label{ex:prelim-lattice}
  Consider the $\pairwise$ program from Example \ref{ex:runnings}.
  In the example, the lattice used in $\pairwise$ is left unspecified, so we consider
  two cases.

  The first lattice we consider is the \emph{discrete} lattice $\discrete(A)$ over a set
  of principals
  $$A = \{ \Alice, \Bob, \Charlie, \ldots \}.$$
  The elements of this lattice is either $\bot$, $\top$, or an element of $A$ and any
  two different elements of $A$ are incomparable: see Figure \ref{fig:discrete}.
  The second lattice we consider is the powerset lattice $\Pow(A)$.

  Naturally, if we consider the input $x = \{1^\Alice, 2^\Bob, 3^\Charlie\}$,
  we get different values for $C(\L(x))$ depending on what lattice the labels
  $\Alice$, $\Bob$, and $\Charlie$ are drawn from.
  %
  %
  In the first lattice, $\discrete(A)$, we have:
  \begin{alignat*}{2}
    C(\L(x)) &= \{\bot, \Alice, \Bob, \Charlie, \top \}\\
    \intertext{While in $\Pow(A)$ we have:}
    C(\L(x)) &= \{\emptyset, \Alice, \Bob, \Charlie,\\
             &\ \ \ \   \{\Alice, \Bob\}, \{\Alice, \Charlie\},\\
             &\ \ \ \   \{\Bob, \Charlie\}, \{\Alice, \Bob, \Charlie\} \}
  \end{alignat*}
  In other words, if $\pairwise$ is implemented with the powerset lattice then
  $\MEF[\pairwise](x)$ runs $\pairwise$ eight times, compared to five when
  $\pairwise$ is implemented with the discrete lattice.
  However, the trade-off is that the powerset lattice allows more fine-grained
  control over security levels; in the discrete lattice all combinations of
  input levels collapse to $\top$, whereas the powerset lattice allows the user
  to see more fine-grained labels for such combined data like $\{\Alice, \Bob\}$.
  In the end of Section \ref{sec:greatAndSmall}, we introduce a family of
  truncated powerset lattices that allows us to fine-tune this trade-off.
\end{example}

Before we dive into how lattice shape influences runtime for multi-execution,
we discuss other possible data representations.
Firstly, we have seen the data representation of Algehed and Flanagan
\cite{MultiExecutionBounds}: inputs and outputs are sets $S \in \Pow(V\times\L)$.
Another possibility is the notion of a faceted value $\Fac$ over some set $V$
\cite{FSME, OptimisingFSME, schmitz2016faceted}:
$$
  f \in \Fac(V) ::= v \in V\ |\ \facet{\ell}{f}{f}
$$
A faceted value is like a decision tree of labels and their meaning can be
given by the selection of a faceted tree at a particular label:
\begin{align*}
  v @ \ell &= v\\
  \facet{\jmath}{f_0}{f_1} @ \ell &=
    \begin{cases}
      f_0 @ \ell\ \ \ \textit{if}\ \jmath \canFlowTo \ell\\
      f_1 @ \ell\ \ \ \textit{otherwise}
    \end{cases}
\end{align*}
With faceted values in mind we can think of a different notion of computation:
$$
p : \text{List}\ \Fac(I) \to \Fac(O)
$$
In this setting, an insecure sum function is:
$$
\badListSum_1(i) = \sum_{f \in i}\ f@\top
$$
It leaks the sum of the most secret view of all its inputs to a public
(non-faceted) output.
Likewise, the following sum function has the same security leak as the
$\badSum$ function above, the output label depends on the presence of labels in
the input:
$$
\badListSum_2(i) = \facet{\bigsqcup \L(i)}{\sum_{f \in i}\ f@\top}{0}
$$
While a good version of the function, that similarly to $\goodSum$ picks a
security-level a-priori, is:
$$
\goodListSum_\ell(i) = \facet{\ell}{\sum_{f \in i}\ f@\ell}{0}
$$

These examples demonstrate that the same kind of functions that one can write
in the $\Pow(I\times\L) \to \Pow(O\times\L)$ setting can be re-created
in the faceted setting.
In fact, for the same reasons one needs to multi-execute for all levels in $C(\L(x))$
in our setting, one needs to multi-execute faceted functions in an analogous
manner \cite{MF, SME, OGMF, OptimisingFSME}.
Consequently, the choice of setting does not decide the overhead of
multi-execution, rather it is still bounded by $|C(\L(x))|$.

\section{Great and Small}
\label{sec:greatAndSmall} 

\begin{table*}
  \setlength{\tabcolsep}{0.5em}
  \renewcommand{\arraystretch}{1.5}
  \caption{\label{tbl:complexity-classes} Complexity Classes for Lattice Constructions}
  \centering
  \begin{tabular}{c|c|c}
    \textbf{Lattice} & \textbf{Complexity} & \textbf{Conditions} \\\hline
    Totally Ordered Naturals $\Nat$           & $\Theta(n)$         & -                   \\\hline
    Discrete Naturals $\discrete(\Nat)$     & $\Theta(n)$         & -                   \\\hline
    Product Lattice $\L_0\times\L_1$ & $\Omega(l_0(\lfloor\frac{n}{2}\rfloor)l_1(\lfloor\frac{n}{2}\rfloor))$
                       $\Oh(u_0(n)u_1(n))$ & $\L_i$ is $\Omega(l_i(n))$ and $\Oh(u_i(n))$. \\\hline
    Vertical sum $\L_0\vsum\L_1$  & $\Theta(f_0(n) + f_1(n))$ & $\L_i$ is $\Theta(f_i(n))$ \\\hline
    Horizontal sum $\L_0\hsum\L_1$  & $\Theta(f_0(n) + f_1(n))$ & $\L_i$ is $\Theta(f_i(n))$ \\\hline
    Exponential Lattice $2^\L$           & $\Theta(2^n)$      &
      $\forall i \not= j.\ \exists \ell_i, \ell_j \in \L.\ \ell_i \not\canFlowTo \ell_j$ \\\hline
    Powerset Lattice $\Pow(A)$       & $\Theta(2^n)$ & $A$ is non-finite \\\hline
    Truncated Powerset Lattice $\Pow_k(A)$     & $\Theta(n^k)$ & $A$ is non-finite \\\hline
    DC Labels (\hspace{-1sp}\cite{DCLabels}) & $\Theta(2^n)$ & The number of principals is non-finite
  \end{tabular}\\
\end{table*}


In this section, we explore how the overhead of multi-execution differs with the
choice of lattice.

\begin{example}
  \label{ex:lattices}
  Following Example \ref{ex:prelim-lattice}, consider the lattice $\discrete(\Nat)$.
  If we take some $S_n \subseteq \discrete(\Nat)$ of size $n$, what is the largest
  we can make $C(S_n)$?
  To answer this, consider some $L \subseteq S_n$, what are the possible values for $\bigsqcup L$?
  It can be only one of three possible things, either $\bigsqcup L = \bot$, $\bigsqcup L = \top$,
  or $\bigsqcup L = i$ for some $i \in \Nat$.
  However, if $\bigsqcup L = i$, then $i \in L$ as the only way to get two elements of $\discrete(\Nat)$
  to join to $i$ is for at least one of them to be $i$ in the first place.
  Consequently, we have that $C(S_n) \subseteq S_n \cup \{ \bot, \top \}$.
  This in turn means that:
  $$
  |C(S_n)| \le |S_n| + 2 = n + 2
  $$
  In other words, $|C(S_n)|$ grows no faster than $n$.

  To contrast, in the lattice $\Pow(\Nat)$ the size of $C(S_n)$ grows more quickly.
  Consider $S_n = \{ \{i\} | i \in [1..n]\}$, the set of singleton sets $\{i\}$
  for $i$ in the interval $1 \le i \le n$.
  The set $S_n$ has size $n$, but the closure set of $S_n$ is much bigger:
  $$C(S_n) = \{ L | L \subseteq [1..n] \}$$
  $C(S_n)$ is the set of all subsets of $[1..n]$ and has size $2^n$.
  In other words, the size of closure sets in $\discrete(\Nat)$ grows linearly
  with the size of the input set, while the size of the closure sets in $\Pow(\Nat)$
  grows exponentially.
\end{example}

We begin to formalise the intuition in Example \ref{ex:lattices} by reminding
the reader of a few notions from complexity theory.
Specifically, Definition \ref{def:bounds} formalizes the key notions of upper and lower bounds.

\begin{definition}[Lower and Upper bounds]
  \label{def:bounds}
  If $f, g : \Nat \to \Nat$ we say that:
  \begin{itemize}
    \item $f(n)$ is $\Oh(g(n))$ ($f(n)$ grows no faster than $g(n)$) if and only if there exists an $N_0 \in \Nat$
    and a $C \in \Q^+$ such that for all $n \ge N_0$ it is the case that
    $f(n) \le g(n)C$.
  \item $f(n)$ is $\Omega(g(n))$ ($f(n)$ grows no slower than $g(n)$) if and only if there exists an
      $N_0 \in \Nat$ and a $C\in \Q^+$ such that for all $n \ge N_0$ it is the
      case that $f(n)
    \ge g(n)C$.
  \item $f(n)$ is $\Theta(g(n))$ ($f(n)$ grows like $g(n)$) if and only if $f(n)$ is $\Oh(g(n))$ and
    $\Omega(g(n))$.
  \end{itemize}
\end{definition}

While standard, Definition \ref{def:bounds} warrants breaking down slightly.
The definition of $f(n)$ being $\Oh(g(n))$ says that there is some point,
$N_0$, after which all $n$ are such that $g(n)$ bigger than or equal to
$f(n)$ up to a constant factor independent of $n$.
Likewise, the definition of $f(n)$ being $\Omega(g(n))$ says that eventually,
as $n \ge N_0$, $f(n)$ is bigger than or equal to $g(n)$ up to a constant factor.
The constant factor provides the generality necessary to allow us to say things
like ``the function $f(n) = 2n^2 + n$ grows like $n^2$'' as it allows us to
formally ignore both the factor $2$ and the addition of $n$.

Next, we translate these bounds from functions to the size of a lattice's
closure sets.
\begin{definition}
  Given a lattice $\L$ define its \emph{closure-size} as:
  $$
  \CS{\L}(n) \defAs \textit{max}\{\ |C(S)|\ |\ S \subseteq\ \L, |S| \le n \}
  $$
\end{definition}

To measure the size of the closures in $\L$, $\CS{\L}(n)$ gives us the size of
the biggest closure set that $\L$ can produce for a set $S \subseteq \L$ of
size at most $n$.
Consequently, $\CS{\L}(n)$ measures the worst-case number of executions of $p$
that $\MEF[p](x)$ does if $|x| = n$ when the given lattice is $\L$.

Next we lift the definition of bounds from functions to lattices to introduce a
convenient terminology for lattices.
\begin{definition}
  The lattice $\L$ is:
  \begin{itemize}
    \item $\Oh(f(n))$ if and only if $\CS{\L}(n)$ is $\Oh(f(n))$.
    \item $\Omega(f(n))$ if and only if $\CS{\L}(n)$ is $\Omega(f(n))$.
    \item $\Theta(f(n))$ if and only if $\CS{\L}(n)$ is $\Theta(f(n))$.
  \end{itemize}
\end{definition}
In Section \ref{sec:fastAndSlow} we will see how these bounds on lattices
translate to worst-case time complexity for multi-execution.
Specifically:
\begin{itemize}
  \item $\Oh(f(n))$ translates to an upper-bound on the worst-case overhead
    of multi-execution, and
  \item $\Omega(f(n))$ translates to a lower-bound on the worst-case overhead, and
  \item $\Theta(f(n))$ gives a tight bound on worst-case overhead.
\end{itemize}

Next we re-visit Examples \ref{ex:prelim-lattice} and \ref{ex:lattices} using
our new terminology.
\begin{example}
  The discrete lattice $\discrete(\Nat)$ is $\Theta(n)$ as we know from Example \ref{ex:lattices}
  that $\CS{\discrete(\Nat)}(n) = n+2$.
  Furthermore, the $\Pow(\Nat)$ lattice is $\Oh(2^n)$ as we know from Example \ref{ex:lattices}
  that $\CS{\Pow(\Nat)}(n) = 2^n$.
  Finally, we note that the linear (total order) lattice $\Nat$ ordered in the
  standard way is $\Theta(n)$, which is demonstrated by the fact that
  $\CS{\Nat}(n) = n+1$ as $C([1\ldots n]) = [0\ldots n]$.
\end{example}

These three bounds have practical implications.
Firstly, both the \texttt{gg} and \texttt{Feature Extraction} lattices
of Alpernas et al. \cite{alpernas2018secure} for describing the security
concerns of multi-user serverless applications are discrete lattices over the
set of users, for which we expect worst-case $\Oh(n)$ overhead.
Secondly, mashup lattices of Magazinius et al.  \cite{magazinius2010lattice}
that allow a website to arbitrarily combine data from third-party domains is a
powerset lattice and we expect worst-case $\Oh(2^n)$ overhead for multi-execution
in their setting.
Finally, the linear lattice $\Nat$ is a generalization of the traditional
``military lattice'' with levels like $\LOW \canFlowTo \MED \canFlowTo \HIGH$
discussed as early as Denning's seminal work introducing lattice-based IFC
\cite{Denning}.

Table \ref{tbl:complexity-classes} summarises the results in this section,
describing the complexity of various lattices, including products
$\times$, two different sum operations $\vsum$ and $\hsum$, exponentiation,
and a few other examples.
For example, if $\L_0$ and $\L_1$ have their closure sets upper bounded
by $u_0(n)$ and $u_1(n)$ respectively, then $\L_0 \times \L_1$ is upper bounded by
$u_0(n)\times u_1(n)$.

Next we provide a set of tools for and examples of how to analyse lattice
shape.
We present a number of basic facts about lattice
shape, and continue to present the analysis that underlies the results in Table
\ref{tbl:complexity-classes}.
Proofs that are not in the body of the paper are found in the appendices.

\begin{restatable}[$\Omega$ families]{Lemma}{lemmaomegafamilies}
  \label{lem:omega-families}
  The lattice $\L$ is $\Omega(f(n))$ if and only if
  there exists a family
  $\{S_n\}_{n\in\Nat}$ of sets such that:
  \begin{enumerate}
    \item $\forall n.\ S_n \subseteq \L$
    \item $\forall n.\ |S_n| \le n$
    \item $|C(S_n)|$ is $\Omega(f(n))$
  \end{enumerate}
\end{restatable}
This lemma gives us the basic building block for proving lower bounds.
An analogous reasoning principle can be established for upper bounds, $\L$ is
$\Oh(f(n))$ if there is a family of sets $L_n$ of size $|L_n| \le f(n)$ such
that $C(S) \subseteq L_n$ for each $S$ of size less than or equal to $n$.
Furthermore, there is a global upper bound on all lattices.
\begin{restatable}[Global Bounds]{Theorem}{theoremglobalbounds}
  \label{thm:global-bounds}
  ~
  \begin{enumerate}
    \item All lattices are $\Oh(2^n)$.
    \item If $\L$ is non-finite, then $\L$ is $\Omega(n)$,
        otherwise it is $\Oh(1)$.
  \end{enumerate}
\end{restatable}
The second item in Theorem \ref{thm:global-bounds} highlights that we treat
all finite lattices the same way, they introduce a constant amount of overhead
in multi-execution.
This is true because our analysis is asymptotic, if the lattice $\L$ is finite
then the maximum number of multi-executions is constant at $|\L|$.
If the lattice is tiny, like the two-point lattice, then treating the overhead
as constant is accurate.
If the lattice is finite but large it is impractical to multi-execute for every
lattice label and the overhead will be dominated by the asymptotic behaviour of
``adaptive'' multi-execution like $\MEF$ or faceted execution.
Finally, we note that potentially unbounded lattices are common in the IFC
literature, e.g. in DC-labels \cite{DCLabels}, the DLM \cite{DLM}, and FLAM
\cite{arden2015flow}.

If the lattice $\L$ is contained in $\L'$, we expect that
$\L'$ is at least as big as $\L$.
To make this formal, we define the notion of a lattice \emph{homomorphism} and
giving us an \emph{embedding}.
\begin{definition}
  A \emph{lattice homomorphism} from $\L$ to $\L'$ is a function
  $h : \L \to \L'$ such that:
  \begin{enumerate}
    \item $h(\bot_\L) = \bot_{\L'}$
    \item $h(\ell \lub \ell') = h(\ell) \lub h(\ell')$
  \end{enumerate}
  If $h$ is injective we say that $h$ is an \emph{embedding} of
  $\L$ in $\L'$ and that $\L$ can be embedded in $\L'$.
\end{definition}
\begin{restatable}[Embedding Complexity]{Theorem}{theoremembedding}
  \label{thm:embedding}
  If $\L$ can be embedded in $\L'$ then:
  \begin{itemize}
    \item If $\L$ is $\Omega(l(n))$ then $\L'$ is $\Omega(l(n))$ and
    \item If $\L'$ is $\Oh(u(n))$ then $\L$ is $\Oh(u(n))$
  \end{itemize}
\end{restatable}
The following example illustrates the usefulness of embeddings
for proving bounds.
\begin{example}
  Free boolean algebras over a set of principals $A$, the set of
  propositional logic formulas with atomic propositions from $A$ ordered by
  implication, form the basis of a number of security lattices in the
  literature, most notably Disjunction Category Labels (DC Labels)
  \cite{DCLabels} and the Flow Limited Authorization Model (FLAM) \cite{FLAM}.
  The powerset lattice can be embedded into any such free boolean algebra
  by the embedding:
  $$\textit{embed}(\{a_0, \ldots, a_n\}) = a_0 \vee \ldots \vee a_1$$
  By Theorems \ref{thm:embedding} and \ref{thm:global-bounds} we now have that the free
  boolean algebra is $\Omega(2^n)$ and $\Oh(2^n)$ respectively, giving us the tight bound of
  $\Theta(2^n)$.
\end{example}
Next, we explore the way that bounds interact with a few methods for
forming lattices from smaller lattices and introduce the \emph{$k$-truncated}
powerset lattice $\Pow_k(A)$.
Specifically, the next three subsections establish results in Table
\ref{tbl:complexity-classes} and the reader is free to skip them on first
reading, while the final subsection is important to understand later examples.

\subsection{Product Lattices}

The first lattice formation method we consider is the product lattice.

\begin{definition}[Product Lattice]
  The lattice $\L_0 \times \L_1$ is called the \emph{product} of lattices $\L_0$ and $\L_1$
  and has pairs $(\ell_0, \ell_1)$ as elements where $\ell_i \in \L_i$ and has
  order:
  $$
  (\ell_0, \ell_1) \canFlowTo (\jmath_0, \jmath_1)\ \iff\ \forall i \in \{0, 1\}.\ \ell_i  \canFlowTo \jmath_i
  $$
\end{definition}

The least-upper-bound of $(\ell_0, \ell_1)$ and $(\jmath_0, \jmath_1)$ is the
pair of least-upper-bounds:
$$
(\ell_0, \ell_1) \lub (\jmath_0, \jmath_1) = (\ell_0 \lub \jmath_0, \ell_1 \lub \jmath_1)
$$

The IFC literature has many examples of product lattices, many of which
simultaneously track both confidentiality and integrity.
For example, an element of the DC-labels lattice \cite{DCLabels} is formed by
taking a pair of CNF formulas over principals; one represents confidentiality
requirements on data and the other integrity requirements.
%

Next we begin to establish bounds for these product lattices.
\begin{restatable}{Theorem}{ohltimesl}
  \label{thm:oh-l-times-l}
  If $\L$ is $\Oh(u(n))$ and $\L'$ is $\Oh(u'(n))$, then
  $\L \times \L'$ is $\Oh(u(n)u'(n))$.
\end{restatable}
This theorem says that upper bounds multiply in the product lattice,
what about lower bounds?
One might expect that if $\L$ and $\L'$ are $\Theta(l(n))$ and $\Theta(l'(n))$
respectively, then $\L\times\L'$ is $\Theta(l(n)l'(n))$.
However, we know that $\Pow(A)$ is $\Theta(2^n)$ for non-finite $A$, if lower
bounds multiply we would have that $\Pow(A)\times\Pow(A)$ is $\Theta(2^n2^n)$.
But $\Pow(A)\times\Pow(A)$, like all lattices, is $\Oh(2^n)$ and consequently
$2^n2^n$ would also be upper bounded by $\Oh(2^n)$, which it is not.

\begin{example}
  The lattice $\Nat\times\Nat$ is $\Theta(n^2)$.
  That $\Nat\times\Nat$ is $\Oh(n^2)$ follows from Theorem \ref{thm:oh-l-times-l}.
  To see that $\Nat\times\Nat$ is $\Omega(n^2)$ we construct the family:
  $$
  S_n = [0\ldots\lfloor \frac{n}{2} \rfloor-1] \times \{0\} \cup \{0\} \times [0\ldots\lfloor \frac{n}{2}\rfloor-1]
  $$
  To see the construction of $S_n$  visually, see Figure \ref{fig:S10}.
  %

  Clearly, $|S_n| \le n$ and so it remains to show that $C(S_n)$ is $\Omega(n^2)$.
  It suffices to show that $[0\ldots\lfloor\frac{n}{2}\rfloor-1]^2 \subseteq C(S_n)$ as
  $|[0\ldots\lfloor\frac{n}{2}\rfloor-1]^2|$ is $\Omega(n^2)$.
  If $(i, j) \in [0\ldots\lfloor\frac{n}{2}\rfloor-1]^2$ then $(i, j) = (i, 0) \lub (0, j)$ and
  $i, j \in [0\ldots\lfloor\frac{n}{2}\rfloor-1]$.
  Consequently, $(i, j) = \bigsqcup\{(i, 0), (0, j)\}$, which is in $C(S_n)$.
  Giving us that $|C(S_n)|$ is $\Omega(n^2)$ and, by Lemma
  \ref{lem:omega-families}, that $\Nat\times\Nat$ is $\Theta(n^2)$.
\end{example}

\begin{figure}
  \center
  \begin{tikzpicture}[scale=0.5]
      \foreach \x in {0,1,...,4}{
      \foreach \y in {0,1,...,4}{
        \node[draw,circle,inner sep=1pt,fill] at (\x,\y+\x) (point-\x-\y) {};
        \draw[] (\x,{max(\y+\x-1,\x)}) -- (point-\x-\y);
        \draw[] ({max(\x-1,0)},{max(\y+\x-1,0)}) -- (point-\x-\y);
      }
    }
    \node[draw,thick,rounded corners=.2cm,dashed,fit=(point-0-0) (point-0-4)] (box-left) {};
    \node[draw,thick,rounded corners=.2cm,dashed, rotate fit = -45, fit =(point-0-0) (point-4-0)] (box-right) {};
    \node[left=0.1cm of box-left] {$[0\ldots\lfloor \frac{10}{2}\rfloor-1] \times \{0\}$};
    \node[right=0.1cm of box-right] {$\{0\} \times [0\ldots\lfloor \frac{10}{2}\rfloor-1]$};
  \end{tikzpicture}
  \caption{\label{fig:S10} An illustration of $S_{10}$.}
\end{figure}
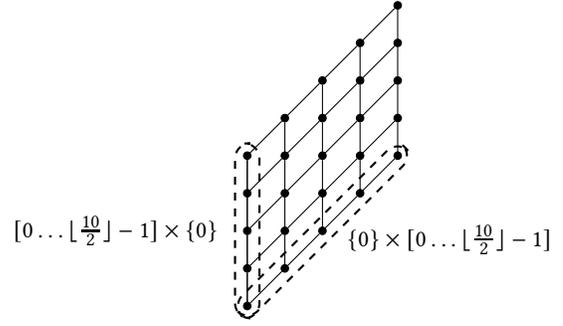

The construction in the example above can be generalised to show that $\Nat^k$ is $\Theta(n^k)$
for any $k$.
The same generalisation allows us to prove the following theorem.

\begin{restatable}{Theorem}{omegaproduct}
  \label{thm:product-lower}
  If $\L$ and $\L'$ are $\Omega(l(n))$ and $\Omega(l'(n))$ for strictly
  positive functions $l(n)$ and $l'(n)$, then $\L\times\L'$ is
  $\Omega(l(\lfloor\frac{n}{2}\rfloor)l'(\lfloor\frac{n}{2}\rfloor))$.
\end{restatable}

Note that the insight that we can divide the two halves of the set
between $\L$ and $\L'$ can be generalise.
Specifically, \emph{any} split between $\L$ and $\L'$ works, which allows us to
give other bounds, such as
$\Omega(l(\lfloor\frac{n}{5}\rfloor)l'(\lfloor\frac{4n}{5}\rfloor))$ and
$\Omega(\textit{maximum}_{0\le k \le n}(l(k)l'(n-k)))$.

This insight may be used to achieve tighter bounds than Theorem \ref{thm:product-lower}.
For example, Theorem \ref{thm:product-lower} gives a lower bound of $\Omega(\frac{n}{2}\cdot2^{\frac{n}{2}})$
for the lattice $\Nat\times\Pow(\Nat)$.
However, taking $S_n = \{0\}\times\{\{i\} | i \in [0..n]\}$ gives the lower-bound $\Oh(2^n)$.

\subsection{Sum Lattices}

\begin{figure}
  \centering
  \begin{tikzpicture}[scale=1]
    \node at (-1, 0) (LL') {$\L\hsum\L'$};
    \node at (0, 0) (EQ) {$=$};
    \node at (1.5, -1) (bot) {$0$};
    \node at (1, 0)  (L)   {$\L$};
    \node at (2, 0)  (L')  {$\L'$};
    \node at (1.5, 1) (top) {$1$};
    \draw (bot) -- (L);
    \draw (bot) -- (L');
    \draw (L) -- (top);
    \draw (L') -- (top);

    \node at (4, 0) (LL'2) {$\L\vsum\L'$};
    \node at (5, 0) (EQ2) {$=$};
    \node at (6, -0.5) (L2) {$\L$};
    \node at (6, 0.5) (L'2) {$\L'$};
    \draw (L2) -- (L'2);
  \end{tikzpicture}
  \caption{\label{fig:latticesums} Sum Lattices}
\end{figure}
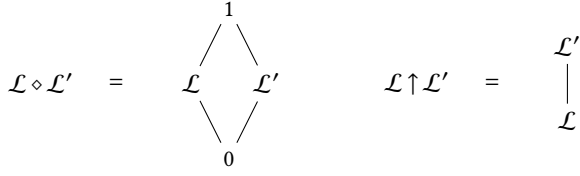

We define two types of sums to capture disconnected parts of a lattice.
\begin{definition}
  The \emph{vertical} sum of $\L$ and $\L'$, written $\L\vsum\L'$, has elements in
  $\L \uplus \L'$ (defined as $\{0\}\times\L \cup \{1\}\times\L'$) where $\ell
  \canFlowTo_{\L\vsum\L'}\ell'$ if and only if:
  \begin{itemize}
    \item $\ell = (0, \jmath)$ and $\ell' = (1,\jmath')$ or
    \item $\ell = (i, \jmath)$, $\ell' = (i, \jmath')$, and $\jmath \canFlowTo \jmath'$
  \end{itemize}

  The \emph{horizontal} sum of $\L$ and $\L'$, written $\L\hsum\L'$,
  has elements in $\{0, 1\} \cup (\L \uplus \L')$
  where $\ell \canFlowTo_{\L\hsum\L'}\ell'$ if and only if either
  $\ell = 0$, $\ell' = 1$, or  $\ell = (i, \jmath)$, $\ell' = (i, \jmath')$, and $\jmath \canFlowTo \jmath'$.
\end{definition}
In other words, $\L\vsum\L'$ is putting $\L'$ on top of $\L$
and $\L\hsum\L'$ is putting $\L$ and $\L'$ next to each other and
gluing $0$ to the bottom and $1$ to the top of the two lattices.
Figure \ref{fig:latticesums} contains a graphical rendition of lattice sums.

These kind of structures appear in the literature in the form of lattices that
incorporate two disjoint parts of an organisation or application.
For example, the \texttt{Hello Retail!} lattice of Alpernas et al. is similar
to a horizontal sum lattice \cite{alpernas2018secure} and the Zone Hierarchies
of Yip et al. also form a horizontal sum \cite{yip2009privacy}.

We can establish bounds on the size of the closure sets for these lattice sums.
Specifically, because there is no complex interaction between $\L$ and $\L'$ in
either $\L\vsum\L'$ nor $\L\hsum\L'$ both sums have closure sets that scale
like the closure sets of $\L$ and $\L'$ taken in isolation.
\begin{restatable}{Theorem}{theoremsumbounds}
  If $\L$ and $\L'$ are $\Theta(f(n))$ and $\Theta(g(n))$ respectively, then $\L\vsum\L'$
  and $\L\hsum\L'$ are both $\Theta(f(n)+g(n))$.
\end{restatable}

\subsection{Exponential Lattices}

\begin{definition}
  \label{def:exp:lattices}
  Given the lattice $\L$ we define the \emph{exponential lattice}
  $2^\L$ as the lattice whose elements are subsets of $\L$
  and where:
  $$\ell \canFlowTo \ell'\ \iff\ \forall \jmath \in \ell.\exists \jmath' \in \ell'.\ \jmath \canFlowTo \jmath'$$
  We also require that the set of labels in $2^\L$ is additionally
  quotiented by the equivalence relation $\sim$ given by:
  $$\ell \sim \ell'\ \iff\ \ell \canFlowTo \ell' \wedge \ell' \canFlowTo \ell$$
\end{definition}

The last requirement of Definition \ref{def:exp:lattices} is a technical
necessity to make $\canFlowTo$ antisymmetric (i.e. that $\ell \canFlowTo \ell'
\canFlowTo \ell \implies \ell = \ell'$ ).
If $\L$ is not quotiented by $\sim$ and there exists $\ell, \ell' \in \L$ such
that $\ell \canFlowTo \ell'$, we have that
$\{\ell'\} \canFlowTo \{\ell, \ell'\} \canFlowTo \{\ell'\}$ but
$\{\ell, \ell'\} \not= \{\ell'\}$.

A label in the $2^\L$ lattice represents the ``most liberal'' extension of
a collection of labels $L \subseteq \L$ to a security label.
It allows us to extend a lattice by introducing additional least upper bounds,
and so it considers more programs secure than the underlying $\L$ lattice.
However, as demonstrated by the following theorem this naturally introduces
additional overhead.
\begin{restatable}{Theorem}{exponentialexponential}
  If there is a non-finite $L \subseteq \L$ such that
  $\forall \ell, \ell' \in L.\ \ell \canFlowTo \ell' \implies \ell = \ell'$
  then $2^\L$ is $\Theta(2^n)$.
\end{restatable}

\subsection{$k$-Truncated Powersets}

The final lattice we explore is the \emph{$k$-truncated powerset lattice}.
This lattice is like the powerset lattice, with the exception that it is
truncated, all sets of size greater than $k$ are replaced
by $\top$.

\begin{definition}
  The lattice $\Pow_k(A)$ is the lattice of all subsets of $A$, ordered by
  inclusion, with cardinality less than or equal to $k$ adjoined with a
  distinguished greatest element $\top$.
\end{definition}

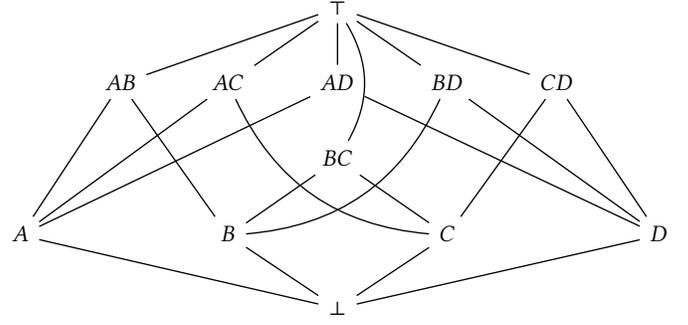
\begin{figure}
  \center
  \begin{tikzcd}
      & & & \top & & & \\
      & AB \arrow[rru, no head] & AC \arrow[ru, no head] & AD \arrow[u, no head] & BD \arrow[lu, no head] & CD \arrow[llu, no head]  & \\

      &  & & BC \arrow[uu, no head, bend right] & & & \\

      A \arrow[ruu, no head] \arrow[rruu, no head] \arrow[rrruu, no head] & &
      B \arrow[luu, no head] \arrow[ru, no head] \arrow[rruu, no head, bend right] & &
      C \arrow[ruu, no head] \arrow[lu, no head] \arrow[lluu, no head, bend left] & &
      D \arrow[luu, no head] \arrow[lluu, no head] \arrow[llluu, no head] \\
      & & & \bot \arrow[lllu, no head] \arrow[lu, no head] \arrow[ru, no head] \arrow[rrru, no head] & & &
  \end{tikzcd}
  \caption{\label{fig:pow2} A graphical representation of the $\Pow_2(\{A, B, C, D\})$ truncated powerset lattice.}
\end{figure}

A graphical rendition of the $2$-truncated powerset lattice can be found in
Figure \ref{fig:pow2}.
There are four principals $A$ through $D$ in the lattice, and a further six
combinations of at most two principals that form the upper bound of any
two singleton labels, for example $AD$ or $BC$.
However, all two-principals labels have upper bound $\top$, unlike the standard powerset
lattice there are no labels $ABC$ or $BCD$ in $\Pow_2(\{A, B, C, D\})$.

To establish bounds on $\Pow_k(A)$ we note that the number of subsets of size $k$
of an $n$ element set is exactly equal to $n$ choose $k$:
$$
{n \choose k} = \frac{n!}{k!(n-k)!}
$$
Two observations about this function are necessary to get convenient upper and
lower bounds for $\Pow_k(A)$\footnote{see Lemma \ref{lem:bin-bounds} in the Appendix}:
$$
\frac{n^k}{k^k} \le \frac{n!}{k!(n-k)!} \le \frac{n^k}{k!}
$$
Because $k$ and consequently $k^k$ and $k!$ are constants, we can
establish a tight bound of $\Theta(n + n^2 + \ldots + n^k) = \Theta(n^k)$ for
the closure-size of $\Pow_k(A)$.

\begin{restatable}{Theorem}{theorempowk}
  If $A$ is non-finite then $\Pow_k(A)$ is $\Theta(n^k)$.
\end{restatable}

\section{Fast and Slow}
\label{sec:fastAndSlow} 

In the previous section, we established bounds for the number of executions
required to multi-execute the program $p$ given the lattice $\L$.
However, it is not sufficient that $\L$ has small closure sets for
multi-execution to be efficient.
Specifically, $\MEF$ has to do a number of things other than executing $p$.
Recall the definition of $\MEF[p]$:
$$\MEF[p](x) = \bigcup\{\ p(x\project\ell)@(\ell\uparrow C(\L(x)))\ |\ \ell \in C(\L(x))\ \}$$
There are three computations that we may classify as overhead:
\begin{enumerate}
  \item Enumerating the elements of $C(\L(x))$.
  \item Computing $x\project\ell$ for each $\ell \in C(\L(x))$.
  \item Computing membership of $\ell \uparrow C(\L(x))$ for $p(x\project\ell)@(\ell\uparrow C(\L(x)))$.
\end{enumerate}
In this analysis, we conservatively assume that all lattice operations (like
$\canFlowTo$ and $\lub$) take constant time.\footnote{For the interested reader
there is a significant body of work on the efficiency of lattice operations
\cite{munro2019space,ait1989efficient,brodnik1999membership}.}
Clearly, (2) above is computable in $\Oh(|x|)$ time.
In this section we show that both (1) and (3) are also computable
with reasonable bounds.

To address (1), the following lemma and theorem allow us to give an algorithm
for efficiently computing (or enumerating) $C(L)$ given that we have an upper
bound on the elements of $C(L)$.

\begin{restatable}{lemma}{lemmadownsetdecide}
\label{lem:downset-decide}
  Let $S\project\ell = \{\ \ell' \in S\ |\ \ell' \canFlowTo \ell\ \}$, then
  $$\bigsqcup (S\project\ell) = \ell\ \iff\ \ell \in C(S)$$
\end{restatable}

\begin{Lemma}
  \label{lem:time-upper-bound}
  Assume a function $f$ that is computable in $\Oh(t(n))$ time and is such that
  for all $S$ it is the case that $C(S) \subseteq f(S)$.
  Then $C(S)$ can be computed in time:
  $$\Oh(t(|S|) + |f(S)||S|)$$
\end{Lemma}
\begin{proof}
  Lemma \ref{lem:downset-decide} suggests a procedure for enumerating $C(S)$
  given $S$ and $f(S)$:
  \begin{algorithmic}
    \FOR {$\ell \in f(S)$}
      \IF {$\bigsqcup (S\project\ell) = \ell$}
        \STATE {\textbf{emit} $\ell$}
      \ENDIF
    \ENDFOR
  \end{algorithmic}
  This procedure takes $\Oh(t(|S|) + |f(S)||S|)$ time and by Lemma
  \ref{lem:downset-decide} we can see that it enumerates $C(S)$.
\end{proof}

Next we tackle (3), Lemma \ref{lem:time-upper-bound} lets us convert our upper
bounds for the elements of $C(S)$ into upper bounds on the time it takes to
compute $C(S)$.
For example, if $S \subseteq \L_0 \times \L_1$ we know that $C(S) \subseteq
C(S_0) \times C(S_1)$ for $S_i = \{\ \ell_i\ |\ (\ell_0, \ell_i) \in S\ \}$ and
so the time it takes to compute $C(S)$ is bounded by $|S||C(S_0)||C(S_1)|$ and
the time it takes to compute $C(S_0) \times C(S_1)$.
Similarly, as we saw in Examples \ref{ex:prelim-lattice} and \ref{ex:lattices} in the discrete
lattice $\discrete(A)$ we have that $C(L) \subseteq L \cup \{\bot, \top\}$.

Furthermore, there is an equivalent formulation of membership in $\ell \uparrow
C(\L(x))$ that can be read as a linear-time algorithm.
This is a novel formulation that does not depend on the size of $C(\L(x))$, in
contrast to the formulation in Section \ref{\LabelPrefix:sec:prelim}.

\begin{Lemma}
  \label{lem:time-uparrow}
  Given $\ell \in C(L)$ it is possible to compute $\jmath \in \ell \uparrow
  C(L)$ in $\Oh(|L|)$ time.
\end{Lemma}
\begin{proof}
  It suffices to check the condition $P(\jmath, \ell)$ defined as:
  $$P(\jmath, \ell)\ \iff\ \ell \canFlowTo \jmath \wedge \forall \iota \in L.\ \iota \canFlowTo \jmath \implies \iota \canFlowTo \ell.$$
  To see that $P(\jmath, \ell) \implies \jmath \in \ell \uparrow C(L)$, consider
  that if $P(\jmath, \ell)$ then $\ell$ is an upper bound on any subset $L'$ of $L$
  such that $\bigsqcup L' \canFlowTo \jmath$ and so $\jmath \in \ell \uparrow C(L)$.
  Likewise, if $\jmath \in \ell \uparrow C(L)$ then $\ell \canFlowTo \jmath$ and if
  $\iota \in L$ and $\iota \canFlowTo \jmath$ then $\bigsqcup \{\iota\} \canFlowTo \jmath$
  and so $\bigsqcup \{\iota\} \canFlowTo \ell$ by the definition of $\jmath \in \ell \uparrow C(L)$.
\end{proof}

Finally, we put these lemmas together to give an upper bound on the execution time of $\MEF[p](x)$.

\begin{restatable}[Time Complexity of Multi-Execution]{Theorem}{theoremtimemef}
  \label{thm:time-MEF}
  Assume:
  \begin{enumerate}
    \item That the lattice $\L$ is $\Oh(s_\L(n))$.
    \item That $p(x)$ can be computed in $\Oh(t_p(|x|))$.
    \item A function $f$ that is computable in $\Oh(t_f(n))$ time and for all $S$,
      $C(S) \subseteq f(S)$ and $|f(S)|$ is $\Oh(s_\L(|S|)$.
  \end{enumerate}
  Then the elements of $\MEF[p](x)$ can be enumerated in time:
  $$\Oh(t_f(|x|) + s_\L(|x|)t_p(|x|)|x|)$$
\end{restatable}

\section{Through the Looking Glass}
\label{sec:lookingGlass} 

If we find that the lattice used by some application is causing unacceptable
performance overheads, how do we switch to a different lattice?
The trick is Galois connections \cite{blyth2005lattices}.

\subsection{Galois Connections}

To understand how Galois connections relate to information flow
control, consider translating data labeled in one lattice
$\L$ to another lattice $\L'$.
To be secure, this needs to be done with a monotonic function $F : \L \to \L'$
that translates each label in $\L$ to a new label in $\L'$.
Suppose additionally that we want to securely back-translate labels from $\L'$
to $\L$ using a function $G : \L' \to \L$.
The goal is then to find $F$ and $G$ such that ``re-labeling'' using $F$ and
``back-labeling'' using $G$ composes to a secure function.

For example, consider the re-labeling function $F$ defined as:
$$
F(\ell) = \text{if}\ \ell \canFlowTo \{\Alice, \Bob\}\ \text{then}\ \bot\ \text{else}\ \top
$$
between the powerset and two-point lattices.
To find a reasonable back-labeling for $F$, consider using $F$ to translate
the labels in the following set:
$$
\{1^\Alice, 2^\Bob, 3^\Charlie, 4^\Dave\} \mapsto_F \{1^\bot, 2^\bot, 3^\top, 4^\top\}
$$
How ought we back-translate this set from the two-point lattice to the powerset lattice?
Back-translating $\bot$ to either $\Alice$ or $\Bob$ would be wrong, as doing so would leak
the value from one to the other.
Likewise, back-translating $\top$ to either $\Charlie$ or $\Dave$ would be
wrong for the same reason.
Fortunately, the following back-translation works well:
$$
G(\ell) = \text{if}\ \ell = \top\ \text{then}\ \top\ \text{else}\ \{\Alice, \Bob\}
$$
The round-trip we get is then:
\begin{align*}
  &\{1^\Alice, 2^\Bob, 3^\Charlie, 4^\Dave\} \mapsto_{F}\\
  &\{1^\bot, 2^\bot, 3^\top, 4^\top\} \mapsto_{G}\\
  &\{1^{\{\Alice,\Bob\}}, 2^{\{\Alice, \Bob\}}, 3^\top, 4^\top\}
\end{align*}

It turns out that if $F$ preserves lower bounds, formally that $F(\bigsqcup L)
= \bigsqcup\{ F(\ell) | \ell \in L \}$, then there is $G$ that is uniquely
determined by $F$ defined as:
$$
G(\jmath) = \bigsqcup \{ \ell \in \L | F(\ell) \canFlowTo_{\L'} \jmath \}
$$
Here we say that $F$ and $G$ form a \emph{Galois connection}.

The usual formal definition of Galois connections is the following.
\begin{definition}
  \label{def:galois}
  A Galois connection $F \dashv G$ between $\L$ and $\L'$ is a pair
  $F : \L \to \L'$ and $G : \L' \to \L$ of functions such that:
  $$F(\ell) \canFlowTo_{\L'} \jmath\ \iff\ \ell \canFlowTo_\L G(\jmath)$$
\end{definition}

Galois connections have a number of useful theoretical properties.
For example, given a Galois connection $F \dashv G$, $G\circ F$ is a
\emph{closure operator}, meaning that:
$$
\ell \canFlowTo G(F(\jmath))\ \iff\ G(F(\ell)) \canFlowTo G(F(\jmath))
$$
Furthermore, two Galois connections $F \dashv G$ between $\L$ and $\L'$ and $F' \dashv G'$ between
$\L'$ and $\L''$ compose to form a Galois connection $(F' \circ F) \dashv (G \circ G')$ between
$\L$ and $\L''$.
The list goes on and the interested reader is encouraged to explore these
structures at their leisure using a textbook of their choice (for example \cite{blyth2005lattices}).

\begin{figure}
  \center
  \begin{tikzcd}
    [execute at end picture={
      \draw[draw=black, rounded corners, thick, dashed]
        (A.north west) -- (AB.north west) -- (AB.north east) -- (B.north east) -- (bot.south east) -- (bot.south west) -- (A.south west) -- cycle;
      \draw[draw=black, rounded corners, thick, dashed]
        (AC.south west) -- (ABC.north west) -- (ABC.north east) -- (BC.north east) -- (C.north east) -- (C.south east) -- (C.south) -- cycle;
      \draw[draw=black, rounded corners, thick, dashed]
        (ABC'.north) -| (ABC'.east) |- (ABC'.south) -| (ABC'.west) |- (ABC'.north);
      \draw[draw=black, rounded corners, thick, dashed]
        (AB'.north) -| (AB'.east) |- (AB'.south) -| (AB'.west) |- (AB'.north);
    }]
    & & |[alias=ABC]|ABC \arrow[ld, no head] \arrow[d, no head] \arrow[rd, no head] & & & & |[alias=ABC']|ABC \arrow[thick, llll, "\unspecify_{AB}", bend right] \\
    & |[alias=AB]|AB & |[alias=AC]|AC \arrow[lld, no head] \arrow[rrd, no head] & |[alias=BC]|BC \arrow[ld, no head] \arrow[rd, no head] \arrow[thick, rrru, "\specify_{AB}", bend left] & & & \\
    |[alias=A]|A \arrow[ru, no head] & {} \arrow[thick, rrrrrd, "\specify_{AB}", bend right] & |[alias=B]|B \arrow[lu, no head] & & |[alias=C]|C &  & \\
    & & |[alias=bot]|\bot \arrow[llu, no head] \arrow[rru, no head] \arrow[u, no head] & &  & & |[alias=AB']|AB \arrow[uuu, no head] \arrow[thick, llllluu, "\unspecify_{AB}", bend left=60]
  \end{tikzcd}
  \caption{\label{fig:galois-select-unselect} A graphical representation of the $\specify_{AB}\dashv\unspecify_{AB}$ Galois connection}
\end{figure}
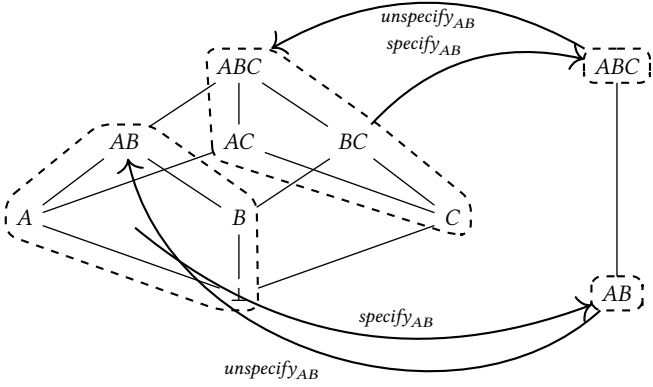

\begin{example}
  \label{ex:galois}
  We can derive a Galois connection between the DC-labels lattice over a set
  $A$ of principals and the $k$-truncated powerset lattice $\Pow_k(A)$.
  The idea is that we map the DC-label $\ell$ to the label of
  at most $k$ principals who either have confidentiality concerns registered in
  or who can vouch for data labeled by $\ell$.
  For example, the label:
  $$(\Alice, \Alice \wedge (\Bob \vee \Charlie))$$
  labels data with $\Alice$'s confidentiality that either $\Alice$ alone or both $\Bob$ and $\Charlie$
  together can vouch for and this label maps to the label $\{\Alice, \Bob, \Charlie\}$ in $\Pow_k(A)$ if $k \ge 3$.
  The label $\{\Alice, \Bob, \Charlie\}$ meanwhile, maps back to the DC-label:
  $$(\Alice \wedge \Bob \wedge \Charlie, \Alice \vee \Bob \vee \Charlie).$$
  The outline of the Galois connection is:
  $$\texttt{DCLabels} \leftrightarrow \Pow(A)\times\Pow(A) \leftrightarrow \Pow(A) \leftrightarrow \Pow_k(A)$$
  The first step maps a DC label to the collection of principals in the label:
  \begin{align*}
    &(\Alice, \Alice \wedge (\Bob \vee \Charlie)) \mapsto\\
    &\;\;\;(\{\Alice\}, \{\Alice, \Bob, \Charlie\})
  \end{align*}
  The second step unites the confidentiality and integrity principals:
  $$
  (\{\Alice\}, \{\Alice, \Bob, \Charlie\}) \mapsto \{\Alice, \Bob, \Charlie\}
  $$
  Finally, we map $\Pow(A)$ to $\Pow_k(A)$ by a Galois connection we call $\truncate_k \dashv \embed$:
  \begin{align*}
    \truncate_k(S) &= \text{if}\ |S| \le k\ \text{then}\ S\ \text{else}\ \top\\
    \embed(S)      &= S
  \end{align*}
  We refer to the $\truncate_k \dashv \embed$ Galois connection as $t\dashv e$ when $k$ is clear from
  the context.
  In our example where $k =3$, truncating the set does nothing.
  
  The chain going back to DC-labels is similar:
  \begin{align*}
    &\{\Alice, \Bob, \Charlie\} \mapsto\\
    &\;\;\;(\{\Alice, \Bob, \Charlie\}, \{\Alice, \Bob, \Charlie\}) \mapsto\\
    &\;\;\;\;\;\;(\Alice \wedge \Bob \wedge \Charlie, \Alice \vee \Bob \vee \Charlie)
  \end{align*}

  If there are more than $k$ principals, like the label
  $\ell = (\Alice \wedge \Bob, \Charlie \vee \Dave)$ for $k = 3$, we lose information when
  going all the way from DC-labels to $\Pow_k(A)$.
  Specifically, $\ell$ maps to $\top$ in $\Pow_k(A)$, which maps back to $\top$.
\end{example}

Next we use the insight that Galois connections can collapse large lattices
into smaller ones by defining a variant of the $\MEF$ enforcement mechanism.
If $F : A \to B$ and $S \subseteq A$ let $F^*(S) = \{\ F(s)\ |\ s \in S\ \}$.
\begin{definition}
  Given a Galois connection $F \dashv G$ between $\L$ and $\L'$,
  we define:
  \begin{alignat*}{2}
    C_{F\dashv G}(S)\       &=\ \mathrlap{G^*(C(F^*(S)))}\\
    \MEF^{F\dashv G}[p](x)\ &=\ \bigcup&&\{\ p(x\project\ell)@(\ell\uparrow C_{F\dashv G}(\L(x)))\\
                            &          &&|\ \ell \in C_{F\dashv G}(\L(x))\ \}
  \end{alignat*}
\end{definition}
Note that $\MEF^{F \dashv G}$ is essentially the same as $\MEF$, except that instead of enumerating
the labels in $C(\L(x))$ we enumerate the labels in $C_{F\dashv G}(\L(x))$.
This means that we only take the closure over the labels in the \emph{target
lattice} of the $F \dashv G$ Galois connection.

Recall the $\goodSum_L$ program from Example \ref{ex:runnings}:
$$
\goodSum_L(x) \defAs \{ (\sum \{\ a\ |\ a^\ell \in x, \ell \in L\ \})^{\bigsqcup L}\ \}
$$
This program is noninterfering and only produces output at level $\bigsqcup L$.
By providing a Galois connection between $\Pow(A)$ and the two-point lattice $\{\bot, \top\}$
we can optimize $\MEF[\goodSum_L]$.
The Galois connection in question, which we call $\specify_\ell \dashv \unspecify_\ell$,
is outlined in Figure \ref{fig:galois-select-unselect} for $\ell = AB$ and is defined as:
\begin{align*}
  \specify_\ell(\jmath) &\defAs\ \text{if}\ \jmath \canFlowTo \ell\ \text{then}\ \bot\ \text{else}\ \top\\
  \unspecify_\ell(\jmath) &\defAs\ \text{if}\ \jmath = \top\ \text{then}\ \top\ \text{else}\ \ell
\end{align*}
The reader is free to verify that these two form a Galois connection.
In the rest of the paper, we abbreviate $\specify_\ell \dashv
\unspecify_\ell$ as $s\dashv u$ when $\ell$ is clear from context.
Now, note that $\bigsqcup L \in C_{\specify_{\bigsqcup L} \dashv
\unspecify_{\bigsqcup L}}(\L(x))$ for all $x$, and so
$$\MEF^{s\dashv u}[\goodSum_L](x) = \MEF[\goodSum_L](x) = \goodSum_L(x)$$
as $\goodSum_L(x\project\bigsqcup L)$ will always be executed by $\MEF^{s\dashv
u}$ and its result included in the final output.
While $\MEF$ needs to run $\Oh(2^n)$ executions of $\goodSum_L$, $\MEF^{s\dashv
u}$ needs one execution of $\goodSum_L$ in the best case, and two in the worst,
as the size of $C_{F\dashv G}(\L(x))$ scales as the target lattice, which in
this case has only two elements.
Specifically, if $\bigsqcup \L(x) \canFlowTo \bigsqcup L$ then $C_{s\dashv
u}(\L(x)) = \{\bigsqcup L\}$ and so there will only be one execution of
$\goodSum_L$, while if there is some $\ell \in \L(x)$ such that $\ell
\not\canFlowTo \bigsqcup L$ we consequently have that $C_{s\dashv u}(\L(x)) =
\{\bigsqcup L, \top\}$ and so we get two executions of $\goodSum_L$.

For a practical example, consider the COWL system \cite{stefan:2014:protecting}
that provides an IFC framework for the web.
One application of COWL is so-called ``mashups'', sites that include content
from many different websites and present it in aggregate.
For example, a mashup can collect and compare price information to display
purchase recommendations from multiple online retailers to the user.
In this mashup, individual retailers need their own security label to manage
sensitive data, and the mashup needs a label that collects data to provide the
price recommendation.

In other words, the mashup site needs the discrete lattice $\discrete(S)$ where
$S$ is the set of sites.
If we use a general label and container system like COWL and attempt to apply
multi-execution to this example, we would need a Galois connection between
DC-labels, the native labels in COWL, and $\discrete(S)$.
Fortunately, $\discrete(S)$ is isomorphic to $\Pow_1(S)$, and so the Galois connection
from Example \ref{ex:galois} with $k = 1$ is sufficient.
This brings the number of executions of our hypothetical site down from
$\Oh(2^n)$ for the DC-labels lattice, to $\Oh(n)$ in the $\Pow_1(S)$ lattice.

We can establish noninterference for $\MEF^{F\dashv G}$.
\begin{restatable}{Theorem}{theoremgaloisnoninterference}
  $\MEF^{F \dashv G}[p]$ is noninterfering.
\end{restatable}

Next we tackle transparency.
The Galois connection changes the behaviour of $\MEF^{F\dashv G}$
and there are Galois connections for which $\MEF^{F\dashv G}[p](x) \not= p(x)$ even
for noninterfering $p$.
In fact, we encourage the reader to come up with an example of such a Galois
connection.\footnote{Hint: what happens if the target lattice has only a single
element?}
However, if a Galois connection accurately captures the behaviour of a
noninterfering program $p$, such as the case in the example of $\goodSum_L$
with $s \dashv u$ above, then we expect that the semantics of $p$ is preserved
by $\MEF^{F\dashv G}$.
\begin{restatable}{Theorem}{theoremgaloistransparent}
  \label{thm:galois-transparency}
  If $p$ is noninterfering and $\jmath \in (G\circ F)^*(\L)$
  then:
  $$\MEF^{F\dashv G}[p](x)@\{\jmath\} = p(x)@\{\jmath\}$$
\end{restatable}
The second precondition of Theorem \ref{thm:galois-transparency} can be read as
stating a condition on $F$ and $G$ relative to $p$.
In effect, it says that a \emph{transparent} Galois connection $F \dashv G$ for the
program $p$ is such that $\forall x.\ \L(p(x)) \subseteq (G\circ F)^*(\L)$.
Consider the noninterfering example programs in Example \ref{ex:runnings},
we see that $\specify_{\bigsqcup L} \dashv \unspecify_{\bigsqcup L}$ is a transparent 
Galois connection for $\goodSum_L$.

However, $\truncate_2 \dashv \embed$ is not a transparent Galois connection for $\pairwise$!
To understand why, consider that:
$$1^{\{\Alice, \Bob, \Charlie\}} \in \pairwise(\{0^{\{\Alice, \Bob\}}, 1^{\{\Charlie\}}\})$$
But $\{\Alice, \Bob, \Charlie\} \not\in (\embed \circ \truncate_2)^*(\Pow(A))$
and consequently $\truncate_2 \dashv \embed$ misses this output of $\pairwise$.
In other words, if there are elements of the input to $\pairwise$ where the labels have
more than one element, this label gets forgotten by the $\truncate_2 \dashv \embed$
Galois connection.
However, this does not preclude this Galois connection from being useful.
For example, we can multi-execute $\pairwise$ on the input $x$ with the
$\truncate_k \dashv \embed$ if $k \ge |\ell \lub \jmath|$ for all pairs of
labels $\ell, \jmath$ in $\L(x)$.

\subsection{Execution Time of $\MEF^{F\dashv G}$}

The factors influencing the execution time of $\MEF$ discussed in Section
\ref{sec:fastAndSlow} are also present for $\MEF^{F\dashv G}$.
Firstly, if $F \dashv G$ goes between $\L$ and $\L'$, the size
of $C_{F\dashv G}(\L(x))$ and so the number of executions of $p$ grows as $\CS{\L'}$.
In other words, if $\L'$ is $\Oh(f(n))$ then the number of executions of $p$
is too.

Secondly, following Lemma \ref{lem:time-uparrow}, next we give a similar
lemma relating to computing membership of $\ell \uparrow C_{F\dashv
G}(\L(x))$ that gives us polynomial time algorithm for this piece of overhead
as well.
\begin{restatable}{Lemma}{theoremgaloistime}
  \label{lem:uptime-galois}
  Given $\ell \in L$ and assuming $F \dashv G$ such that both $F$ and $G$ are
  constant time it is possible to compute $\jmath \in G(F(\ell)) \uparrow C_{F\dashv G}(L)$
  in $\Oh(|L|)$ time by computing:
  $$
    G(F(\ell)) \canFlowTo \jmath \wedge\ \forall \iota \in F^*(L).\ G(\iota) \canFlowTo \jmath \implies G(\iota) \canFlowTo G(F(\ell))
  $$
\end{restatable}
With this lemma, we have all the pieces we need to find the execution time
bound on $\MEF^{F\dashv G}[p](x)$.
It is essentially the same as the bound in Theorem \ref{thm:time-MEF}, where
the lattice $\L$ is the target lattice of the $F \dashv G$ Galois connection.

\begin{restatable}[Time Complexity of Multi-Execution]{Theorem}{theoremtimemefgalois}
  \label{thm:time-MEF-Galois}
  Assume:
  \begin{enumerate}
    \item That $F \vdash G$ is a Galois connection between $\L$ and $\L'$.
    \item That the lattice $\L'$ is $\Oh(s_\L(n))$.
    \item That $p(x)$ can be computed in $\Oh(t_p(|x|))$.
    \item A function $f$ that is computable in $\Oh(t_f(n))$ time and for all $S$,
      $C(S) \subseteq f(S)$ and $|f(S)|$ is $\Oh(s_\L(|S|)$ for the lattice $\L'$.
  \end{enumerate}
  Then the elements of $\MEF^{F\vdash G}[p](x)$ can be enumerated in time:
  $$\Oh(t_f(|x|) + s_\L(|x|)t_p(|x|)|x|)$$
\end{restatable}

\subsection{Finding Galois Connections}

Next we address the issue of finding the right Galois connection.
Specifically, we show that there is a specification for a \emph{most coarse
grained} Galois connection for each program.
This can be used to show that, for some programs, no transparent Galois connection can
reduce overhead while for other programs.
We do this by using the closure operator ($G \circ F$) of a Galois connection
$F \dashv G$.
Specifically, we obtain a Galois connection from any closure operator $k$.
\begin{definition}
  If $\L$ is a lattice we say that a function $k : \L \to \L$ is a \emph{closure
  operator} if and only if it satisfies:
  \begin{alignat*}{2}
    &\text{(1) Extensivity:}\;  &&\ell \canFlowTo k(\ell)\\
    &\text{(2) Monotonicity:}\; &&\ell \canFlowTo \jmath \implies k(\ell) \canFlowTo k(\jmath)\\
    &\text{(3) Idempotence:}\;  &&k(\ell) = k(k(\ell))
  \end{alignat*}
  Alternatively, $k$ is a closure operator if and only if
  $$\ell \canFlowTo k(\jmath)\ \iff\ k(\ell) \canFlowTo k(\jmath).$$
\end{definition}
\begin{Theorem}[From \cite{blyth2005lattices}]
  If $k$ is a closure operator then $\L\setminus k$ forms a lattice
  with equivalence classes up to $k$ for elements and the order inherited
  from $\L$.
\end{Theorem}
\begin{Theorem}[From \cite{blyth2005lattices}]
  \label{thm:closure-galois} 
  If $k$ is a closure operator then $k \dashv \lambda \ell. \ell$ is a Galois
  connection between $\L$ and the quotient lattice $\L\setminus \{(\ell,
  \jmath)|k(\ell) = k(\jmath)\}$, where elements are equivalence classes up to
  $k$ ordered by $\canFlowTo_\L$ on representative elements (fixpoints of $k$).
\end{Theorem}
The key corollary of Theorem \ref{thm:closure-galois} is that if $k$
differentiates between all labels that are in the range of $p$, then $k$ gives
rise to a transparent Galois connection that can be used together with
$\MEF^{F\dashv G}$ to multi-execute $p$.
\begin{corollary}
  Given a closure operator $k$ on $\L$ such that $\L(p(x)) \subseteq k^*(\L)$
  for all $x$ we have that $k \dashv \lambda \ell. \ell$ is a transparent
  Galois connection between $\L$ and $\L\setminus \{(\ell, \jmath)|k(\ell) = k(\jmath)\}$ for $p$
\end{corollary}
Now, we can finally give a detailed specification of a closure operator
for lattices that have greatest lower bounds.
\begin{definition}
  \label{def:kp}
  Given program $p : \Pow(I \times \L) \to \Pow(O \times \L)$ for a lattice $\L$ with meets
  we define:
  $$
  k_p(\ell) = \bigsqcap\{\ \jmath\ |\ \exists x.\ \jmath \in \L(p(x)) \wedge \ell \canFlowTo \jmath\ \}
  $$
\end{definition}
\begin{restatable}{Theorem}{theoremkpclosure}
  \label{thm:kpclosure}
  $k_p$ is a closure operator.
\end{restatable}
Using Theorem \ref{thm:kpclosure} we can re-construct the transparency
of the $\specify_{\bigsqcup L} \dashv \unspecify_{\bigsqcup L}$ Galois connection
for $\goodSum_L$.
Specifically, $k_{\goodSum_L}$ is isomorphic to $\specify_{\bigsqcup L}$:
$$
k_{\goodSum_L}(\ell) = \text{if}\ \ell \canFlowTo \bigsqcup L\ \text{then}\ \bigsqcap\{\bigsqcup L\}\ \text{else}\ \bigsqcap\emptyset
$$

Having seen that the $k_p$ closure operators allow us to
compute Galois connections from programs, we next turn to the question of
universality.
Specifically, we prove that the $k_p$ closure operator is the most coarse-grained
closure operator that gives rise to a transparent Galois connection.
This coupled with the fact that each Galois connection can be determined up to
isomorphism from some closure operator and vice versa \cite{blyth2005lattices}
then gives a roadmap to determine for a given program if there is some Galois
connection that allows for efficient multi-execution.
\begin{restatable}[Canonicity of $k_p$]{Lemma}{lemmacanonicity}
  \label{lem:canonicity_kp}
  Given $k$ such that $\L(p(x)) \subseteq k^*(\L)$ for all $x$,
  if $k_p(\ell) \not= k_p(\jmath)$ then $k(\ell) \not= k(\jmath)$.
\end{restatable}
A consequence of this Lemma is that the overhead for multi-execution on
$\pairwise$ can not be reduced below $\Oh(2^n)$ for the powerset lattice in the
worst-case using a Galois connection to a smaller lattice.
Specifically, this is because
$\{\ \ell\ |\ \exists x.\ \ell \in \L(\pairwise(x)) \} = \L$
and so $k_\pairwise(\ell) = \ell$ for all $\ell$.
Consequently, if $\L = \Pow(A)$ then every transparent Galois connection $F \dashv G$ for
$\pairwise$ has to preserve the structure of the powerset lattice, and thus
preserving the $\Oh(2^n)$ bound on $C(\L(x))$ in $C_{F\dashv G}(\L(x))$.

However, the following variant of $\pairwise$ that only works on singleton or
empty labels admits a transparent Galois connection:
$$
\pairwise_1(x) \defAs \{\ \textit{max}(a, b)^{\ell \lub \jmath}\ |\  a^\ell \in x, b^\jmath \in x, |\ell| = |\jmath| \le 1\ \}
$$
Specifically, we can construct the set $S$ of all labels in its co-domain:
$$
S = \{\ \ell\ |\ \exists x.\ \ell \in \L(\pairwise_1(x)) \} = \{\ \ell\ |\ |\ell| \le 2\ \}
$$
From which we get:
$$
k_{\pairwise_1}(\ell) = \text{if}\ |\ell| \le 2\ \text{then}\ \ell\ \text{else}\ \top = \truncate_2(\ell)
$$

To formalise this reasoning, the final theorem of this Section shows that every
transparent Galois connection for $p$ introduces at least as many executions in
$\MEF^{F\dashv G}$ as the Galois connection given by $k_p$.
\begin{restatable}[Canonicity of Galois Connections]{Theorem}{theoremcanonicitytwo}
  If $F \dashv G$ is a transparent for $p$, then $|C_{F\dashv G}(L)| \ge |k_p^*(C(L))|$.
\end{restatable}

With this theory in place, we see a clear path for future work to take our
analysis of lattice shape and bring it into practice.
One recipe for harnessing this theory is the following:
\begin{enumerate}
  \item Propose a procedure $K(p)$ to approximate $k_p$.
  \item Find the complexity of $K$ and of the resulting lattice.
  \item Show that the overhead of $K(p)$ and $\MEF^{K(p) \dashv \lambda \ell. \ell}[p](x)$
    is less than that of $\MEF[p](x)$.
\end{enumerate}
This opens up a new research direction for Galois-Multi-Execution that we hope
will incorporate insights from across the static and dynamic program analysis
literature.

\section{Empirical Results}
\label{sec:empirical}

\begin{figure}
  \centering
\begin{tikzpicture}[scale=0.8]
  \begin{semilogyaxis}[ title = {$\MEF$ Performance}
                      , xlabel = {Input Size}
                      , ylabel = {Execution Time ($s$) (log axis)}
                      , legend pos = north east 
                      , legend cell align = left
                      , legend entries={$\MEF[\goodSum_L]$, $\MEF^{s\dashv u}[\goodSum_L]$, $\goodSum_L$}]
    \addplot[only marks, mark = x] table [x=Size, y=Mean, col sep=comma] {data/mef-goodSum.csv};
    \addplot[only marks, mark = triangle, mark repeat = 5] table [x=Size, y=Mean, col sep=comma] {data/mefGalois-goodSum.csv};
    \addplot[only marks, mark = o, mark repeat = 5] table [x=Size, y=Mean, col sep=comma] {data/goodSum.csv};
    \addplot[no markers, opacity = 0.5] gnuplot [raw gnuplot] { 
      set datafile separator ",";
      f(x) = b*x + c;     
      b=0.00001;c=0.0001;          
      fit f(x) 'data/mefGalois-goodSum.csv' u 1:2 via b,c; 
      plot [x=10:100] f(x);    
      set print "fit-parameters/mefGalois-goodSum-b.tex"; 
      print b; 
      set print "fit-parameters/mefGalois-goodSum-c.tex"; 
      print c; 
    };
    \addplot[no markers, opacity=0.5] gnuplot [raw gnuplot] { 
      set datafile separator ",";
      f(x) = b*x+c;     
      b=0.00001;c=0.0001;          
      fit f(x) 'data/goodSum.csv' u 1:2 via b,c; 
      plot [x=10:100] f(x);    
      set print "fit-parameters/goodSum-b.tex"; 
      print b; 
      set print "fit-parameters/goodSum-c.tex"; 
      print c; 
    };
  \end{semilogyaxis}
\end{tikzpicture}
  \caption{\label{fig:performance:sum:compare} The Execution Times of $\MEF$ and $\MEF^{s\dashv u}$ on $\goodSum_L$}
\end{figure}
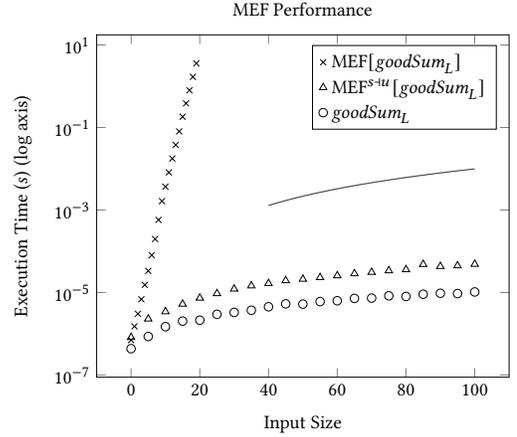

To validate our theoretical results empirically, we have implemented the
framework used in this paper as a
small\footnote{350 lines of code, including all our experiments} Haskell \cite{Haskell}
library\footnote{That is available as supplementary material to this paper}.
Our enforcement mechanisms $\MEF$ and $\MEF^{F\vdash G}$ are implemented as
higher order functions:
\begin{verbatim}
      mef :: (Lattice l, Ord l, Ord b)
          => (Set (a, l) -> Set (b, l))
          ->  Set (a, l) -> Set (b, l)

mefGalois :: ( Lattice l, Lattice l', Ord l, Ord l', Ord b)
          => Galois l l'
          -> (Set (a, l) -> Set (b, l))
          ->  Set (a, l) -> Set (b, l)
\end{verbatim}
The type signature for \texttt{mef} comes in three parts, line by line:
\begin{enumerate}
  \item \texttt{(Lattice l, Ord l, Ord b)} are constraints that require \texttt{l} to be a type that forms
    a lattice (\texttt{Lattice l}) that additionally has a total order (\texttt{Ord l}), note that this does
    not require the lattice ordering $\canFlowTo$ on \texttt{l} to be total, and that
    \texttt{b} has a total order (\texttt{Ord b}).
    The total order constraints are necessary in order to efficiently represent the inputs and outputs of
    \texttt{mef} as sets (implemented as e.g. AVL or red-black trees).
  \item \texttt{(Set (a, l) -> Set (b, l))} is a higher-order argument, a
    function \texttt{p} that takes sets of labeled \texttt{a}s as input and
    produces labeled \texttt{b}s as output.
  \item \texttt{Set (a, l) -> Set (b, l)} means that \texttt{mef p} is also a
    function from \texttt{Set (a, l)} to \texttt{Set (b, l)}.
\end{enumerate}
The difference between \texttt{mef} and \texttt{mefGalois} is that
\texttt{mefGalois} additionally requires two lattices and a Galois connection
between them as input (where a Galois connection is a pair of functions).

Figure \ref{fig:performance:sum:compare} contains a teaser of our empirical results.
It contains the log of runtime for $\MEF[\goodSum_L]$, $\MEF^{\specify \dashv
\unspecify}[\goodSum_L]$, and $\goodSum_L$ plotted against $100$ inputs of size
ranging from $0$ to $100$.
As can be seen in the Figure, going from using the powerset lattice to the
two-point lattice reduces the running time from exponential to polynomial.

In the experiments an input of size $n$ is the set
$\{1^{\{p_1\}}\ldots n^{\{p_n\}}\}$ where each principal $p_i$
is unique.
The definition of $L$ for $\goodSum_L$ for input $n$ is $L = \{p_1 \ldots p_n\}$.
$\MEF^{s\dashv u}$ uses the Galois connection $\specify_{\{p_1 \ldots p_n\}} \dashv \unspecify_{\{p_1 \ldots p_n\}}$
where the ``specified'' element of the powerset lattice is precisely
$\{p_1 \ldots p_n\}$.

From Figure \ref{fig:performance:sum:compare} we see that $\MEF[\goodSum_L]$ takes
exponential time.
$\MEF^{s\dashv u}$, meanwhile, has linear-time performance.
This is because $\MEF^{s\dashv u}$ introduces at most two executions of each
program, one for $\{p_1, \ldots,p_n\}$ and one for $\top$, and so the
running-time is proportional to the running time of $\goodSum_L$.
Figure \ref{fig:performance:sum:compare} also contains linear fit lines for
$\goodSum_L$ and $\MEF^{s\dashv u}[\goodSum_L]$.

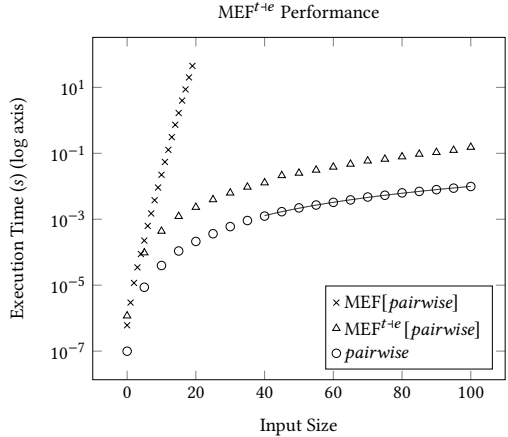
\begin{figure}
  \centering
\begin{tikzpicture}[scale=0.8]
  \begin{semilogyaxis}[ title = {$\MEF^{t\dashv e}$ Performance}
              , xlabel = {Input Size}
              , ylabel = {Execution Time ($s$) (log axis)}
              , scaled y ticks={base 10:2}
              , legend pos = south east
              , legend cell align = left ]
    \addplot[only marks, mark = x] table [x=Size, y=Mean, col sep=comma] {data/mef-pairwise.csv};
    \addlegendentry{$\MEF[\pairwise]$}
    \addplot[only marks, mark = triangle, mark repeat = 5] table [x=Size, y=Mean, col sep=comma] {data/mefGalois-pairwise.csv};
    \addlegendentry{$\MEF^{t\dashv e}[\pairwise]$}
    \addplot[only marks, mark = o, mark repeat = 5] table [x=Size, y=Mean, col sep=comma] {data/pairwise.csv};
    \addlegendentry{$\pairwise$}
    \addplot[no markers, opacity = 0.5] gnuplot [raw gnuplot] { 
      set datafile separator ",";
      f(x) = a*x**2+b*x + c;     
      a=0.00001;b=0.00001;c=0.0001;          
      fit f(x) 'data/mefGalois-pairwise.csv' u 1:2 via a,b,c; 
      plot [x=40:100] f(x);    
      set print "fit-parameters/mefGalois-pairwise-a.tex"; 
      print a; 
      set print "fit-parameters/mefGalois-pairwise-b.tex"; 
      print b; 
      set print "fit-parameters/mefGalois-pairwise-c.tex"; 
      print c; 
    };
    \addplot[no markers, opacity=0.5] gnuplot [raw gnuplot] { 
      set datafile separator ",";
      f(x) = a*x**2+b*x+c;     
      a=0.00001;b=0.00001;c=0.0001;          
      fit f(x) 'data/pairwise.csv' u 1:2 via a,b,c; 
      plot [x=40:100] f(x);    
      set print "fit-parameters/pairwise-a.tex"; 
      print a; 
      set print "fit-parameters/pairwise-b.tex"; 
      print b; 
      set print "fit-parameters/pairwise-c.tex"; 
      print c; 
    };
  \end{semilogyaxis}
\end{tikzpicture}
  \caption{\label{fig:performance:pairwise} The Execution Times of $\MEF$ and $\MEF^{t\dashv e}$ on $\pairwise$}
\end{figure}

In a second experiment, reported in Figure \ref{fig:performance:pairwise},
we compare the execution time of $\MEF$ and $\MEF^{t\dashv e}$ on the $\pairwise$
function.
The $t\dashv e$ Galois connection here refers to the canonical Galois connection
between $\Pow(A)$ and $\Pow_2(A)$.
The input is the same as in the other experiment.
We also fit the timing results for running $\MEF^{t\dashv e}[\pairwise]$ and
$\pairwise$ to polynomials of the shape $ax^2+bx+c$  using the Gnuplot
\cite{gnuplot} implementation of the Levenberg-Marquardt algorithm
\cite{levenberg1944method}.
The results can be seen as fit lines in Figure \ref{fig:performance:pairwise}.
The fits are good when the input size is greater than $40$, showing us that
quadratic order execution time is a good asymptotic fit.


\section{Related Work}
\label{sec:related}

There is a growing body of work on transparent IFC enforcement
\cite{SME,MF,jaskelioff2011secure,FSME,OGMF,OptimisingFSME,MultiExecutionBounds,
zanarini2013precise,de2012flowfox,de2014secure,ngo2018impossibility,
ngo2015runtime,pfeffer2019efficient,micinski2020abstracting,
rafnsson2016secure,bolocsteanu2016asymmetric,bielova2016spot} and the
efficiency of multi-execution as an enforcement mechanism has been studied in
various settings going back to the beginning \cite{SME}.

\paragraph*{Theoretical Work on Multi-Execution Performance}
In the introduction of the first paper on multi-execution Devriese and Piessens
\cite{SME} remark that:
\begin{quote}
``One obvious disadvantage of multi-execution is its cost in terms of CPU time
and memory use.''
\end{quote}
Some theoretical effort has gone into using the multiple-facets \cite{MF}
framework for reducing the number of \emph{superfluous} runs of parts of
programs under multi-execution \cite{OGMF,OptimisingFSME}.
Algehed et al. \cite{OptimisingFSME} limit the number of runs of a program
under multi-execution by pruning what output levels in the security lattice are
used by providing a primitive for pruning the lattice using a boolean algebra
over labels.
Their modified transparency criteria is similar to our notion of
transparency up to a Galois connection in Theorem \ref{thm:galois-transparency}. 
The $\Omega(|C(\L(x))|)$ lower bound on the overhead of black-box transparent
enforcement has been informally discussed in the literature
\cite{MF,FSME,OptimisingFSME,OGMF}, and was recently formally proven \cite{MultiExecutionBounds}.

\paragraph*{Empirical Work on Multi-Execution Performance}
The first empirical measurements of the performance overhead of multi-execution
are in the original paper, where the authors study the timing overhead of SME
for the two-point lattice on a number of small but realistic benchmarks
\cite{SME}.
Additionally, the ``real-world'' overhead of SME has been studied in the
setting of the FlowFox IFC browser \cite{de2012flowfox}.
The first experiment to study how multi-execution scales with the number of
security levels that we are aware of was in the work on the
multiple facets (MF) version of multi-execution \cite{MF}.
Extending MF, Schmitz et al. \cite{FSME} present Faceted Secure multi-execution
(FSME) that unifies MF and SME and study trade-offs between time and memory use
in these two formulations of multi-execution.
Finally, Algehed et al. \cite{OptimisingFSME} empirically evaluate the effect
of filtering the views (akin to executions in multi-execution) of faceted
values that appear in MF by selecting executions that may lead to observable
outputs and ignoring ones that do not and find that it presents similar speedups
to the ones presented in this paper.

\paragraph*{Implementations of IFC in Haskell}
There is a significant body of work on embedding IFC in Haskell
\cite{MAC,LIO,jaskelioff2011secure,FSME,DCCInHaskell,waye2017cryptographically,
HLIO,parker2019lweb,giffin2012hails,vassena2017securing,vassena2019foundations}.
Most of which falls into the category of ``monadic'' IFC libraries in which the
code that is subject to IFC enforcement is written using a specialized
interface exported by the library.
This differs from our implementation, which works on non-monadic, native, code.
On the other hand, because our setting only applies to batch-job programs, our
library does not deal with reactive program IO, non-determinism, and other such
effects.
However, we believe that this limitation is orthogonal to how lattice shape
influences overhead even in more fully-fledged multi-execution implementations,
like FlowFox \cite{de2012flowfox} and Multef \cite{FSME}.

\section{Conclusions}
In this paper, we have presented a framework for reasoning about upper and
lower bounds on the time overhead of multi-execution.
We have shown that the choice of lattice alters this overhead; lattices that
allow the programmer to express many different combinations of security levels
result in large overheads.
We also show how to use Galois connections to switch between different
lattices, thus allowing programmers to switch from a lattice with high overhead
to one with low overhead.
This switching potentially comes at the cost of altering the behaviour of the
target program, but for many programs it is possible to reduce overhead without
affecting semantics.
We show that a canonical Galois connection that is both as coarse-grained as
possible (reducing the overhead as much as possible) and does not alter program
semantics exists for every lattice with greatest lower bounds.
Finally, we empirically evaluate our performance predictions on a small
implementation of our framework in Haskell and find that the theory matches our
empirical results.

%% file: appendix.tex
\appendixsection{Great and Small}
\label{app:greatAndSmall} 
\begin{restatable}{Lemma}{lemmauseful}
  \label{lem:useful}
  ~
  \begin{enumerate}
    \item $f(n)$ is $\Oh(g(n))$ if and only if $g(n)$ is $\Omega(f(n))$
    \item If $f(n)$ is $\Theta(g(n))$, then $g(n)$ is $\Theta(f(n))$
    \item If $f$ and $g$ are everywhere positive, then $\maxi(f(n),g(n))$ is $\Theta(f(n) + g(n)$.
  \end{enumerate}
\end{restatable}
\begin{proof}
    For (1) and (2) consider that if $f(n)$ is $\Oh(g(n))$ then for some $N_f
    \ge 0$ and $C_f > 0$ we have that for all $n \ge N_f$, $f(n) \le g(n)C$,
    but then $g(n) \ge f(n)C_f^{-1}$ and so $g(n)$ is $\Omega(f(n))$
    with $N_0 = N_f$ and $C = C_f^{-1}$ and vice verse.
    For (3) we have the bounds $\frac{1}{2}(f(n) + g(n)) \le \maxi(f(n),g(n)) \le f(n) + g(n)$.
\end{proof}

\begin{Lemma}[Closure Monotonicity]
  $\CS{\L}$ is monotonic.
\end{Lemma}
\begin{proof}
  For each $n$ there is an $S_n$ such that $|S_n| \le n$
  and $\CS{\L}(n) = |C(S_n)|$, consider $m \ge n$, then
  $|S_n| \le n < m$ and so $\CS{\L}(n) = |C(S_n)| \le \CS{\L}(m)$.
\end{proof}

\lemmaomegafamilies*
\begin{proof}
  Assume a family $\{S_n\}_{n\in\Nat}$ with the
  required properties.
  By definition, $\CS{\L}(n) \ge |C(S_n)|$.
  Therefore $\CS{\L}(n)$ is $\Omega(|C(S_n)|)$.
  Consequently, $\CS{\L}(n)$ is also $\Omega(f(n))$.
  For the other direction, assume $\CS{\L}(n)$ is $\Omega(f(n))$, then by
  the definition of $\CS{\L}$ we have that for each $n$ there is a 
  $S_n \subseteq \L$ of size $n$ such that $|C(S_n)| = \CS{\L}(n)$.
  As $\CS{\L}(n)$ is $\Omega(f(n))$, so is $|C(S_n)|$, and so
  $S_n$ is the required family.
\end{proof}

\theoremglobalbounds*
\begin{proof}
  ~
  \begin{enumerate}
    \item For all $S \subseteq \L$, $|\{\ S'\ |\ S' \subseteq S \}| = 2^{|S|}$,
      so $|C(S)| \le 2^{|S|}$, consequently $\CS{\L}(n)$ is $\Oh(2^n)$.
    \item If $\L$ is non-finite observe that for each $S = \L$
  we have that $S \subseteq C(S)$ and so $|S| \le |C(S)|$ and because $\L$
  is non-finite for each $n \in \Nat$ there exists an $S_n \subseteq \L$ such
  that $|S_n| = n$.
  This defines a family $\{S_n\}_{n\in\Nat}$ that is $\Omega(n)$, consequently
  by Lemma \ref{lem:omega-families} so is $\L$.
  If $\L$ is finite observe that $0 < \CS{\L}(n) \le |\L|$ and so we pick $C =
  |\L|$ and $N_0 = 0$, we get that for all $n \ge 0 = N_0$, $\CS{\L}(n) \le
  |\L| = C = 1C$.
  \end{enumerate}
\end{proof}

\theoremembedding*
\begin{proof}
  If $h : \L \to \L'$ is an embedding of $\L$ in $\L'$, then clearly
  $h(C(S)) = C(h(S))$ as $h$ preserves joins.
  Furthermore, $h$ is injective and so $|h(S)| = |S|$ for all $S$.
  This means that $\CS{\L}(n) \le \CS{\L'}(n)$ as for every $S \subseteq \L$
  such that $|S| = n$ and $|C(S)| = m$, it is the case that
  $h(S) \subseteq \L'$, $|h(S)| = n$, and $|C(h(S))| = |h(C(S))| = |C(S)| = m$.
  From this the required bounds follow trivially.
\end{proof}

\begin{Lemma}
  \label{lem:lubs-distr}
  If $L = \{(\ell_1, \ell'_1),\ldots,(\ell_k,\ell'_k)\}$ where
  $\ell_i \in \L$ and $\ell'_i \in \L'$ for all $i$, then
  $\bigsqcup L = (\bigsqcup \{\ell_1,\ldots,\ell_k\},\bigsqcup \{\ell'_1,\ldots,\ell'_k\})$.
\end{Lemma}
\begin{proof}
  We prove this by induction on $k$.
  In the case when $k = 0$, we have that
  $\bigsqcup \emptyset = \bot = (\bot, \bot) = (\bigsqcup \emptyset, \bigsqcup \emptyset)$.
  In the case when $k = k_0 + 1$ we have, by the induction hypothesis, that
  $L_0 = \{(\ell_1, \ell'_1),\ldots,(\ell_{k_0},\ell'_{k_0})\}$ and
  $\bigsqcup L_0 = (\bigsqcup \{\ell_1,\ldots,\ell_{k_0}\},\bigsqcup \{\ell'_1,\ldots,\ell'_{k_0}\})$.
  If $L = L_0 \cup \{(\ell_k, \ell'_k)\}$ we have that
  \begin{align*}
    &\bigsqcup L =\\
    &L \lub (\ell_k, \ell'_k) =\\
    &(\bigsqcup \{\ell_1,\ldots,\ell_{k_0}\},\bigsqcup \{\ell'_1,\ldots,\ell'_{k_0}\}) \lub (\ell_k, \ell'_k) =\\ 
    &(\bigsqcup \{\ell_1,\ldots,\ell_{k}\},\bigsqcup \{\ell'_1,\ldots,\ell'_{k}\}),
  \end{align*}
    which completes the proof.
\end{proof}

\omegaproduct*
\begin{proof}
  By Lemma \ref{lem:omega-families} we have that both $\L$ and $\L'$ admit
  families $S_n$ and $S'_n$ where $|C(S_n)|$ and $|C(S'_n)|$ are
  $\Omega(l(n))$ and $\Omega(l'(n))$ respectively.
  As a consequence, $|C(S_{\lfloor{\frac{n}{s}}\rfloor})|$ and $|C(S'_{\lfloor{\frac{n}{s}}\rfloor})|$
  are $\Omega(l(\lfloor{\frac{n}{s}}\rfloor))$ and $\Omega(l'(\lfloor{\frac{n}{s}}\rfloor))$
  respectively.
  Let $Z_n = S_{\lfloor{\frac{n}{s}}\rfloor}$ and $Z'_n = S_{\lfloor{\frac{n}{s}}\rfloor}$.
  We now construct the family $P_n = Z_n\times\{\bot_\L'\}\cup \{\bot_\L\}\times Z'_n$
  in $\L\times\L'$.
  It is the case that $|P_n| = |Z_n| + |Z'_n| \le 2\lfloor \frac{n}{2} \rfloor \le n$.
  Finally, it remains to show that $C(P_n) \supseteq C(Z_n)\times C(Z'_n)$, which gives us the lower
  bound that $|C(P_n)|$ is $\Omega(l(\lfloor{\frac{n}{s}}\rfloor)l'(\lfloor{\frac{n}{s}}\rfloor))$
  as both $l(n)$ and $l'(n)$ are strictly positive functions.
  Let $k = \lfloor\frac{n}{2}\rfloor$, if $(\ell, \ell') \in C(Z_n)\times C(Z'_n)$ then by the definition of closure sets
  $(\ell, \ell') = (\bigsqcup L, \bigsqcup L')$ for $L = \{\ell_1,\ldots,\ell_k\} \subseteq Z_n$
  and $L' = \{\ell'_1,\ldots,\ell'_k\} \subseteq Z'_n$.
  By Lemma \ref{lem:lubs-distr} we have that
  $\bigsqcup\{(\ell_1,\ell'_1),\ldots,(\ell_k,\ell'_k)\} = (\bigsqcup L, \bigsqcup L') = (\ell, \ell')$.
  Finally, $\{(\ell_1,\ell'_1),\ldots,(\ell_k,\ell'_k)\} \in P_n$ and so $(\ell, \ell') \in C(P_n)$.
  Consequently, $C(P_n) \supseteq C(Z_n)\times C(Z'_n)$ and so $|C(P_n)|$ is 
  $\Omega(l(\lfloor{\frac{n}{s}}\rfloor)l'(\lfloor{\frac{n}{s}}\rfloor))$.
\end{proof}

\ohltimesl*
\begin{proof}
  We show that for all:
  $$S = \{(\ell_1, \ell'_1),\ldots,(\ell_n, \ell'_n)\} \subseteq \L\times\L'$$
  We have that:
  $$C(S) \subseteq C(\{\ell_1,\ldots,\ell_n\}) \times C(\{\ell'_1,\ldots,\ell'_n\})$$
  Let $L = \{\ell_1,\ldots,\ell_n\}$ and $L' = \{\ell'_1,\ldots,\ell'_n\}$.
  Now pick any $(\ell, \ell') \in C(S)$, by definition of $C(S)$ we have that:
  $$(\ell, \ell') = \bigsqcup\{(\jmath_1, \jmath'_1),\ldots,(\jmath_k,\jmath'_k)\}$$
  For some $J = \{\jmath_1,\ldots,\jmath_k\} \subseteq L$
  and $J' = \{\jmath'_1,\ldots,\jmath'_k\} \subseteq L'$.
  Furthermore, by Lemma \ref{lem:lubs-distr} we have that $(\ell, \ell') = (\bigsqcup J, \bigsqcup J')$.
  This gives us $\ell = \bigsqcup J \in C(L)$
  and $\ell' = \bigsqcup J' \in C(L')$.
  In other words, $(\ell, \ell') \in C(L) \times C(L')$ and so
  $C(S) \subseteq C(L) \times C(L')$.
  This immediately lets us conclude that:
  $$
  \CS{\L\times\L'}(n) \le \CS{\L}(n)\CS{\L'}(n)
  $$
  Giving us that $\L\times\L'$ is $\Oh(u(n)u'(n))$.
\end{proof}

\begin{Lemma}
  \label{lem:sum-lower-bound}
  If $\L$ and $\L'$ are $\Omega(f(n))$ and $\Omega(g(n))$ respectively, then
  $\L\vsum\L'$ and $\L\hsum\L'$ are both $\Omega(f(n) + g(n))$.
\end{Lemma}
\begin{proof}
  Call $S_n$ and $S'_n$ the respective $\Omega$-families of $\L$ and $\L'$ given
  by Lemma \ref{lem:omega-families}.
  We construct the family $F_n$ by, for each $n$, picking $F_n = \{0\}\times S_n$
  if $|C(S_n)| > |C(S'_n)|$ and $F_n = \{1\}\times S'_n$ otherwise.
  Clearly, $|C(F_n)| \ge \maxi(|C(S_n)|,|C(S'_n)|)$ in both $\L\vsum\L'$ and
  $\L\hsum\L'$ and so $|C(F_n)|$ is $\Omega(\maxi(f(n),g(n)))$.
  By Lemmas \ref{lem:useful} and \ref{lem:omega-families} we have that
  both $\L\vsum\L'$ and $\L\hsum\L'$ are $\Omega(f(n)) + g(n))$.
\end{proof}

\begin{Lemma}
  \label{lem:sum-upper-bound}
  If $\L$ and $\L'$ are $\Oh(f(n))$ and $\Oh(g(n))$ respectively, then
  $\L\vsum\L'$ and $\L\hsum\L'$ are both $\Oh(f(n) + g(n))$.
\end{Lemma}
\begin{proof}
  For $\L\vsum\L'$ consider that $S \subseteq \L\vsum\L'$ means that
  there exists $L \subseteq \L$ and $L' \subseteq \L'$ such that
  $S = L \uplus L'$.
  Next we show that $C(S) \subseteq C(L) \uplus C(L')$, by observing that
  for any $J \subseteq L \uplus L'$ there are two cases, either there are
  elements $(1, \ell') \in J'$ such that $\ell' \in L'$ or there are not.
  In the first case, $\bigsqcup J = (1, \bigsqcup \{ \ell'\ |\ (1, \ell') \in J \})$,
  as $(0, \ell) \lub (1, \ell') = \ell'$ and so $\bigsqcup J \in \{1\}\times C(L') \subseteq C(L) \uplus C(L')$.
  In the second case, $\bigsqcup J = (0, \bigsqcup \{ \ell\ |\ (0, \ell) \in J \})$
  and so $\bigsqcup J \in \{0\}\times C(L) \subseteq C(L) \uplus C(L')$.
  Because $C(L \uplus L') \subseteq C(L) \uplus C(L')$ we also have that
  $|C(L \uplus L')| \le |C(L) \uplus C(L')| = |C(L)| + |C(L')|$.
  Because $\CS{\L}$ and $\CS{\L'}$ are both monotnoe functions, this
  means that if $|L \uplus L'| \le n$ then $|C(L)| + |C(L')| \le \CS{\L}(n) + \CS{\L'}(n)$,
  which gives us our bound on $\CS{\L\vsum\L'}(n)$ of $\Oh(f(n) + g(n))$.

  To see that $\L\hsum\L'$ is $\Oh(f(n) + g(n))$ observe
  that if $S \subseteq \L\hsum\L'$ then $C(S) \subseteq \{0,1\} \cup (C(L) \uplus C(L'))$
  for $L$ and $L'$ such that $S = L \uplus L' \cup J$ for $J \subseteq \{0, 1\}$.
  A similar observation as above then immediately gives the upper bound.
\end{proof}

\theoremsumbounds*
\begin{proof}
  Follows immediately from Lemmas \ref{lem:sum-lower-bound} and
  \ref{lem:sum-upper-bound}.
\end{proof}

\exponentialexponential*
\begin{proof}
  We construct the $\Omega$-family $S_n = \{\{\ell_1\}, \ldots, \{\ell_n\}\}$ where
  $\ell_i \in L$ and $\ell_i = \ell_j \implies i = j$.
  Clearly, $C(S_n)$ is the powerset of $\{\ell_1,\ldots,\ell_n\}$ and so has size
  $2^n$, using Lemma \ref{lem:omega-families} this gives us our lower bound on $2^\L$
  of $\Omega(2^n)$.
  Together with Theorem \ref{thm:global-bounds} we get that $\L$ is $\Theta(2^n)$.
\end{proof}

\begin{Lemma}
  \label{lem:bin-bounds}
  $$\frac{n^k}{k^k} \le {n \choose k}$$
\end{Lemma}
\begin{proof}
  For $k = 1$ we have ${n \choose k} = n = \frac{n^k}{k^k}$.
  For $k > 1$ and $0 < x < k \le n$ we have:
  $$
    \frac{n-x}{k-x} - \frac{n}{k}\ =\ \frac{xn-xk}{k(k-x)}\ \ge\ 0
  $$
  Giving us that $\frac{n-x}{k-x} \ge \frac{n}{k}$ and hence:
  $$\frac{n^k}{k^k} \le \frac{n}{k}\cdot\frac{n-1}{k-1}\cdot\ldots\cdot\frac{n-k+1}{1} = {n \choose k}$$ 
\end{proof}

\theorempowk*
\begin{proof}
  In the case when $k = 0$ $\Pow_k(A)$ is a finite lattice consisting of
  $\emptyset$ and $\top$ and so it is $\Theta(1)$.
  In the case when $k > 0$ we prove the upper and lower bound separately.
  For the upper bound, without loss of generality consider any
  $S = \{\ell_1, \ldots, \ell_n\}$ such that no $\ell_i = \emptyset$,
  it is the case that:
  $$
  C(S) \subseteq \{\top\} \cup \bigcup_{0\le i \le k} \{\ \bigsqcup S'\ |\ S' \subseteq S, |S'| = i \}
  $$
  In other words, each element of $C(S)$ is either $\top$, or a set of size at
  most $k$ that can be constructed by taking the suprenum of some $S' \subseteq S$ of size
  at most $k$.
  Consequently, we get the following inequality for the size of $C(S)$:
  $$
  |C(S)| \le 1 + \sum_{0\le i \le k} |\{\ S'\ |\ S' \subseteq S, |S'| = i\ \} \le 1 + \sum_{0\le i \le k} |S|^i
  $$
  Which in turn means that $\Pow_k(A)$ is $\Oh(n^k)$.
  For the lower bound, assume $A$ is non-finite and let $S_n = \{ a_1, \ldots, a_n \}$.
  Such that all $a_i$ are distinct, which gives us:
  $$
  C(S_n) \supseteq \bigcup_{0\le i \le k} \{\ \bigsqcup S'\ |\ S' \subseteq S_n, |S'| = i \}
  $$
  Giving us:
  $$
  |C(S_n)| \ge \sum_{0 \le i \le k} {n \choose i} \ge 1 + \sum_{1\le i \le k} \frac{n^i}{i^i} \ge \frac{n^k}{k^k}
  $$
  Which is sufficient to establish that $\Pow_k(A)$ is $\Omega(n^k)$, and so $\Pow_k(A)$ is $\Theta(n^k)$.
\end{proof}

\appendixsection{Fast and Slow}
\label{app:fastAndSlow} 
\lemmadownsetdecide*
\begin{proof}
  Left to right is trivial, $S\project\ell \subseteq S$ and so $\bigsqcup S\project\ell \in C(S)$.
  For the right to left direction, consider that if $\ell \in C(S)$ then there exists
  an $L \subseteq S$ such that $\bigsqcup L = \ell$.
  Note that $L \subseteq S\project\ell$ and that if $L' \subseteq S\project\ell$ then
  $\bigsqcup L \canFlowTo \bigsqcup (L \cup L') \canFlowTo \ell$ as $\ell$ is an upper-bound
  of $L'$ and $\bigsqcup$ is monotone with respect to $\cup$.
  Consequently, we have that $\ell = \bigsqcup L \canFlowTo \bigsqcup (L \cup (S\project\ell - L)) \canFlowTo \ell$
  which means that $\ell \canFlowTo S\project\ell \canFlowTo \ell$, giving us $S\project\ell = \ell$.
\end{proof}

\theoremtimemef*
\begin{proof}
  Firstly, Lemma \ref{lem:time-upper-bound} means it takes $t_f(n) + |f(\L(x))|n$ to enumerate $C(\L(x))$.
  Secondly, there are $\Oh(s_\L(n)t_p(n))$ runs of $p$, each of which produces outputs
  bounded in size by $\Oh(t_p(n))$ and for each such output Lemma \ref{lem:time-uparrow}
  gives the time taken to compute membership of the respective $\uparrow$-set as $\Oh(n)$.
  Thirdly, $|f(\L(x))|$ is $\Oh(s_\L(\L(x)))$ and so can be ignored.
  Finally, putting these bounds together gives the time taken to enumerate all
  elements $\ell \in C(\L(x))$, computing $p(x\project\ell)$ and filtering them
  by $\ell\uparrow C(\L(x))$.
\end{proof}

\appendixsection{Through the Looking Glass}
\label{app:lookingGlass}
\begin{Lemma}
  \label{lem:C-galois}
  If $L \sim_\ell L'$ then $C_{F \dashv G}(L) \sim_\ell C_{F \dashv G}(L')$.
\end{Lemma}
\begin{proof}
  Consider $\jmath \canFlowTo \ell$ such that $\jmath \in C_{F \dashv G}(L)$.
  There exists an $S \subseteq L$ such that $\jmath = G(\bigsqcup F^*(S))$.
  Because $F \dashv G$, $G$ and $G\circ F$ are both monotone and so
  for any $\imath \in S$ it is the case that:
  $$\imath \canFlowTo G(F(\imath)) \canFlowTo G(\bigsqcup F^*(S)) = \jmath$$
  This means that $\jmath$ is an upper bound of $S \subseteq L$, which by
  $L \sim_\ell L'$ means that $S \subseteq L'$.
  In other words, $C_{F \dashv G}(L)\project\ell \subseteq C_{F \dashv G}(L')\project\ell$.
  By symmetry of $\sim_\ell$ we have that $C_{F \dashv G}(L')\project\ell \subseteq C_{F \dashv G}(L)\project\ell$
  and so $C_{F \dashv G}(L) \sim_\ell C_{F \dashv G}(L')$.
\end{proof}

\theoremgaloisnoninterference*
\begin{proof}
  Consider, $\ell$, $x \sim_\ell y$ and $a^\jmath$ such that $a^\jmath \in \MEF^{F \dashv G}[p](x)\project\ell$.
  The definitions of $@$ and $\MEF^{F \dashv G}$ give us that there is some
  $\iota \in C_{F \dashv G}(\L(x))$ such that $\iota \canFlowTo \jmath$ and $a^\jmath \in p(x\project\iota)@(\iota\uparrow C_{F \dashv G}(\L(x)))$.
  However, because $\iota \canFlowTo \jmath \canFlowTo \ell$ we also have that $\iota \in C_{F \dashv G}(\L(y))$,
  by Lemma \ref{lem:C-galois}.
  Furthermore, if $\jmath \not\in \iota\uparrow C_{F \dashv G}(\L(y))$ then there is some other $\iota'\in C_{F \dashv G}(\L(y))$
  such that $\iota' \canFlowTo \jmath$ and $\iota' \not\in C_{F \dashv G}(\L(y))$, but this is impossible because $x \sim_\ell y$
  and so $a^\jmath \in p(y\project\iota)@(\iota\uparrow C_{F \dashv G}(\L(y)))$.
  By definition of $\MEF^{F \dashv G}$ this gives us $a^\jmath \in \MEF^{F \dashv G}[p](y)\project\ell$
  and so $\MEF^{F \dashv G}[p](x)\project\ell \subseteq \MEF^{F \dashv G}[p](y)\project\ell$.
  Symmetry of $\sim_\ell$ means that
  $\MEF^{F \dashv G}[p](y)\project\ell \subseteq \MEF^{F \dashv G}[p](x)\project\ell$
  and so $\MEF^{F \dashv G}[p]$ is noninterfering.
\end{proof}

\begin{Lemma}
  \label{lem:C-galois-covers-C}
  Take a Galois connection $F \dashv G$ between $\L$ and $\L'$ and a set $L \subseteq \L$.
  For any $\jmath \in (G \circ F)^*(\L)$ there is some $\ell \in C_{F\dashv G}(L)$ such that
  $\jmath \in \ell \uparrow C_{F\dashv G}(L)$.
\end{Lemma}
\begin{proof}
  Let $\ell = G(\bigsqcup F^*(L\project\jmath))$, $\ell$ is in $C_{F\dashv G}(L)$.
  We have that $\ell = G(F(\bigsqcup(L\project\jmath)))$ and $\bigsqcup(L\project\jmath) \canFlowTo \jmath$
  and so $\ell \canFlowTo G(F(\jmath)) = \jmath$.
  Additionally, if $\ell' \in C_{F\dashv G}(L)$ and $\ell' \canFlowTo \jmath$
  then for some $L'$ it is the case that $G(\bigsqcup F^*(L')) = G(F(\bigsqcup L')) \canFlowTo \jmath$
  and consequently $\bigsqcup L' \canFlowTo G(F(\bigsqcup L')) \canFlowTo \jmath$ and so
  $\bigsqcup L' \canFlowTo \bigsqcup (L\project\jmath)$ and so $\ell' \canFlowTo \ell$.
  Giving us $\jmath \in \ell \uparrow C_{F\dashv G}(L)$.
\end{proof}

\theoremgaloistransparent*
\begin{proof}
  \begin{align*}
       &\MEF^{F\dashv G}[p](x)@\jmath\\
       &\ \ \ \ (\text{Definition})\\
    =\ &\bigcup\{\ p(x\project\ell)@(\ell\uparrow C_{F\dashv G}(\L(x)))\ |\ \ell \in C_{F\dashv G}(\L(x))\ \}@\jmath\\
       &\ \ \ \ (\text{(1) and $x\project\ell = x\project\jmath$ for $\jmath \in \ell \uparrow C_{F\dashv G}(\L(x))$})\\
    =\ &\bigcup\{\ p(x)@(\ell\uparrow C_{F\dashv G}(\L(x)))\ |\ \ell \in C_{F\dashv G}(\L(x))\ \}@\jmath\\
       &\ \ \ \ (\text{(2), Lemma \ref{lem:C-galois-covers-C}, and Definition of $\uparrow$})\\
    =\ &p(x)@\jmath
  \end{align*}
\end{proof}

\theoremgaloistime*
\begin{proof}
  Call the condition in the Lemma statement $P(\jmath, \ell)$.
  To see that $P(\jmath, \ell) \implies \jmath \in \ell \uparrow C_{F\dashv G}(L)$, consider
  that if $P(\jmath, \ell)$ and some $\iota = G(\bigsqcup L') \in C_{F\dashv G}(L)$ is
  such that $\iota \canFlowTo \jmath$ then either $L' = \emptyset$
  in which case $\iota \canFlowTo G(F(\ell))$ by monotonicity of $G$, or
  all $\ell'$ such that $F(\ell') \in L'$ are such that $G(F(\ell')) \canFlowTo \iota$ by monotonicity of $G$
  giving $G(F(\ell')) \canFlowTo \jmath$ and so $G(F(\ell'))) \canFlowTo G(F(\ell))$ by $P(\jmath, \ell)$.
  However, $\iota' = F(\ell')$ for some $\ell'$ and so we really have $G(F(\ell')) \canFlowTo G(F(\ell))$
  for all $\ell'$ such that $F(\ell') \in L'$.
  Consequently, $\ell' \canFlowTo G(F(\ell))$ (by $G\circ F$ being a closure operator)
  and so $F(\ell') \canFlowTo F(\ell)$ (as $F \circ G \circ F = F$ for any
  Galois connection).
  Consequently, $\bigsqcup L' \canFlowTo F(\ell)$ which by monotonicity means that
  $\iota = G(\bigsqcup L') \canFlowTo G(F(\ell))$.
  For the other direction, if $\jmath \in G(F(\ell)) \uparrow C_{F\dashv G}(L)$ then clearly
  $G(F(\ell)) \canFlowTo \jmath$ and if $\iota \in F^*(L)$ and $G(\iota) \canFlowTo \jmath$
  then $G(\iota) \canFlowTo G(F(\ell))$ by the
  definition of $\jmath \in G(F(\ell)) \uparrow C_{F\dashv G}(L)$.
\end{proof}

\theoremtimemefgalois*
\begin{proof}
  Firstly, Lemma \ref{lem:time-upper-bound} means it takes $t_f(n) + |f(\L(x))|n$ to enumerate $G^*(C(F^*(\L(x))))$.
  Secondly, there are $\Oh(s_\L(n)t_p(n))$ runs of $p$, each of which produces outputs
  bounded in size by $\Oh(t_p(n))$ and for each such output Lemma \ref{lem:uptime-galois}
  gives the time taken to compute membership of the respective $\uparrow$-set as $\Oh(n)$.
  Thirdly, $|f(\L(x))|$ is $\Oh(s_\L(\L(x)))$ and so can be ignored.
  Finally, putting these bounds together gives the time taken to enumerate all
  elements $\ell \in C(\L(x))$, computing $p(x\project\ell)$ and filtering them
  by $\ell\uparrow C(\L(x))$.
\end{proof}

\theoremkpclosure*
\begin{proof}
  We have three proof obligations:
  \begin{alignat*}{2}
    &\text{(1) Extensivity:}\;  &&\ell \canFlowTo k_p(\ell)\\
    &\text{(2) Monotonicity:}\; &&\ell \canFlowTo \jmath \implies k_p(\ell) \canFlowTo k_p(\jmath)\\
    &\text{(3) Idempotence:}\;  &&k_p(\ell) = k_p(k_p(\ell))
  \end{alignat*}
  First we let $S_p(\ell) = \{\ \jmath\ |\ \exists x.\ \jmath \in \L(p(x)) \wedge \ell \canFlowTo \jmath\ \}$
  and note that $k_p(\ell) = \bigsqcap S_p(\ell)$.
  Proof obligations in order:
  \begin{enumerate}
    \item $\ell$ is a lower bound of $S_p(\ell)$ and so
      $\ell \canFlowTo \bigsqcap S_p(\ell) = k_p(\ell)$.
    \item If $\ell \canFlowTo \jmath$ then $S_p(\ell) \supseteq S_p(\jmath)$
      and so $k_p(\ell) \canFlowTo k_p(\jmath)$.
    \item $S_p(k_p(\ell)) = S_p(\ell)$ and so $k_p(\ell) = k_p(k_p(\ell))$.
  \end{enumerate}
\end{proof}

\lemmacanonicity*
\begin{proof}
  If $k_p(\ell) \not= k_p(\jmath)$ then without loss of generality we can assume that
  there is some $\iota \in \L(p(x))$ for some $x$ such that $\ell \canFlowTo \iota$
  but $\jmath \not\canFlowTo \iota$.
  Consequently, $\ell \lub \jmath \not\canFlowTo \iota$.
  Assume for a contradiction that $k(\ell) = k(\jmath)$.
  By monotonicity of $k$ we know that $\ell \canFlowTo k(\ell)$ and $\jmath \canFlowTo k(\ell)$
  and so $\ell \lub \jmath \canFlowTo k(\ell)$.
  However, monotonicty of $k$ also gives us that because $\ell \canFlowTo \iota$ we have that $k(\ell) \canFlowTo k(\iota)$.
  But $\iota \in \L(p(x))$ for some $x$ so $k(\iota) = \iota$ and so $k(\ell) \canFlowTo \iota$.
  Putting everything together gives us $\jmath \canFlowTo \ell \lub \jmath \canFlowTo k(\ell) \canFlowTo \iota$
  which contradicts $\jmath \not\canFlowTo \iota$, so $k(\ell) \not= k(\jmath)$.
\end{proof}

\theoremcanonicitytwo*
\begin{proof}
  We have that $\L(p(x)) \subseteq (G\circ F)^*(\L)$ for all $x$ as $F \vdash G$
  is transparent.
  By Lemma \ref{lem:canonicity_kp} we have that:
  $$k_p(\ell) \not= k_p(\jmath) \implies G(F(\ell)) \not= G(F(\jmath))$$
  Therefore, any two elements in $C(L)$ that are distinguished by $k_p$ are distinguished by $G\circ F$.
  Consequently, as:
  \begin{align*}
       &C_{F\vdash G}(L)\ =\ \{ G(\bigsqcup F^*(S) | S \subseteq L \}\ =\\
       &\{ G(F(\bigsqcup S)) | S \subseteq L \}\ =\ (G\circ F)^*(C(L))
  \end{align*}
  this means that $C_{F\vdash G}(L)$ has at
  least as many elements as $k_p^*(C(L))$.
\end{proof}